\documentclass[10pt]{article}

\usepackage[margin=1in]{geometry}
\usepackage{setspace}
\usepackage{graphicx}
\usepackage{booktabs}
\usepackage{adjustbox}
\usepackage{caption} %

\usepackage[hyphens]{url}
\usepackage{hyperref}
\usepackage[round]{natbib}

\usepackage{amsmath}
\usepackage{amsfonts}
\usepackage{amssymb}
\usepackage{amsthm}
\usepackage{dsfont}
\newcommand{\ind}{\mathds{1}}

\usepackage{tikz}
\usetikzlibrary{quantikz2}

\usepackage{xcolor}

\newtheorem{theorem}{Theorem}
\newtheorem{proposition}{Proposition}
\newtheorem{lemma}{Lemma}
\newtheorem{corollary}{Corollary}
\theoremstyle{remark}

\newcommand{\UQ}{\mathbb{U}_{\mathrm{Q}}}
\newcommand{\Uref}{U_{\mathrm{ref}}}
\newcommand{\cN}{\mathcal{N}}
\newcommand{\cF}{\mathcal{F}}
\newcommand{\cM}{\mathcal{M}}
\newcommand{\cD}{\mathcal{D}}
\newcommand{\cK}{\mathcal{K}}
\newcommand{\cI}{\mathcal{I}}
\newcommand{\cH}{\mathcal{H}}
\newcommand{\cE}{\mathcal{E}}
\newcommand{\Rad}{\mathfrak{R}}
\newcommand{\bbR}{\mathbb{R}}
\newcommand{\bbC}{\mathbb{C}}
\newcommand{\snorm}[1]{\left\lVert #1 \right\rVert_\sigma}
\newcommand{\Otil}{\widetilde{O}}

\newcommand{\tightertab}{\footnotesize\setlength{\tabcolsep}{3pt}}

\newcommand{\revtable}[1]{#1}
\newcommand{\revcolor}{}

\title{Paper Title Goes Here}

\begin{document}

\begin{center}

{\Large\bfseries
Margin-Based Generalization for Quantum Machine Learning with Classical Data Encoding
\par}

\vspace{0.9em}

{\large
Ju-Young Ryu\textsuperscript{1,*}
\quad
Tak Hur\textsuperscript{2,3,*}
\quad
Daniel K. Park\textsuperscript{2,4,5}
\par}

\vspace{0.6em}

{\small
\textsuperscript{1}Korea Institute of Science and Technology Information, Republic of Korea \\
\textsuperscript{2}Department of Statistics and Data Science,
  Yonsei University, Republic of Korea\\
\textsuperscript{3} Department of Applied Mathematics, 
University of Waterloo, Canada\\
\textsuperscript{4}Department of Applied Statistics,
  Yonsei University, Republic of Korea\\
  \textsuperscript{5}Department of Quantum Information,
  Yonsei University, Republic of Korea\\
\par}

\vspace{0.4em}

\begingroup
\renewcommand{\thefootnote}{*}
\footnotetext{These authors contributed equally to this work.}
\endgroup

\vspace{0.6em}

{\small
\texttt{jyryu@kisti.re.kr}
\quad
\texttt{tak.hur@uwaterloo.ca}
\quad
\texttt{dkd.park@yonsei.ac.kr}
\par}

\vspace{1.2em}

\end{center}

\begin{abstract}
\noindent
Understanding generalization is essential for assessing the practical viability of quantum machine learning. In many applications, quantum models operate on classical data that must first be encoded into quantum states, making the encoding part of the effective hypothesis class. We develop a margin-based generalization framework for this setting, covering both fixed and trainable quantum data encodings. For fixed encodings, the existing margin bound transfers unchanged. For trainable encodings, we derive a bound using a two-scale covering argument that combines the original margin-dependent term with an additive encoding-capacity term of order \(\widetilde O(\sqrt{L/m})\), where \(L\) is the number of trainable encoding parameters and \(m\) is the training-sample size. Across 19,440 simulated training runs spanning various image datasets, circuit architectures, and optimization settings, the inverse margin performs best among five baseline and two bound-derived complexity measures under conditional-independence and 90th-percentile sign-error criteria, within bootstrap uncertainty. These results provide a unified margin-based view of generalization in quantum machine learning with classical data. 
\end{abstract}

\vspace{0.8em}

\section{Introduction}
\label{sec:intro}

Quantum neural networks (QNNs) use parameterized quantum circuits as trainable models for learning tasks. As in classical machine learning, their practical utility depends on \emph{generalization}: the ability to make accurate predictions on unseen data from a finite training sample. Generalization in QML has been studied through parameter counts \citep{caro2022generalization}, effective dimension \citep{abbas2021power}, and quantum information-theoretic quantities \citep{banchi2021generalization}.

Most of these analyses produce \emph{uniform} bounds that hold simultaneously for every hypothesis in the model class. Uniform bounds have faced sustained criticism in classical deep learning. \citet{zhang2021understanding} showed that deep networks easily fit random labels, so any bound that applies equally to all hypotheses in the class cannot distinguish models that generalize from models that memorize. The same failure occurs in QML. \citet{gilfuster2024understanding} demonstrated that even few-qubit QNNs can perfectly fit randomized labels on a benchmark quantum dataset, rendering uniform bounds equally uninformative in the quantum setting.

Classical learning theory responded to this critique with \emph{margin-based} analyses \citep{cortes1995support,bartlett2017spectrally,neyshabur2018pacbayesian}. These analyses replace the worst case over the hypothesis class with a data-dependent quantity that measures how confidently the trained model separates the training sample. Large-scale empirical studies subsequently confirmed that margin-based complexity measures are among the most reliable predictors of generalization in deep networks \citep{jiang2019fantastic,dziugaite2020search}. \citet{hur2025margins} recently extended this line of work to QML. They established a margin-based generalization bound for QNNs classifying \emph{quantum} data and showed empirically, on quantum phase recognition tasks, that margin-based metrics predict generalization better than parameter-based metrics.

Near-term applications of QML, however, predominantly involve \emph{classical} data, which must be loaded into a quantum computer by a data-encoding circuit. The encoding is known to shape the expressivity of the resulting model \citep{schuld2021effect,perezsalinas2020data} and its generalization behavior \citep{caro2021encoding}. When the encoding is \emph{fixed}, each classical input is mapped to a quantum state once and for all. The classical-data problem then reduces exactly to the quantum-data one, and the margin theory of \citet{hur2025margins} transfers unchanged (Proposition~\ref{prop:fixed}). Modern QNN architectures, however, often make the encoding itself \emph{trainable}, either by interleaving parameterized layers with data re-uploading blocks \citep{perezsalinas2020data,lloyd2020quantum} or by learning the embedding with a classical neural network \citep{hur2024neural}. In this case, the encoded sample matrix depends on learned parameters, so the fixed-matrix covering argument of \citet{hur2025margins} no longer applies directly.

This paper makes the following theoretical and empirical contributions.
\begin{itemize}
\item We show that for any fixed data encoding the quantum-data margin bound holds without modification (Proposition~\ref{prop:fixed}). We then prove a margin-based generalization bound for QNNs whose inputs are produced by a trainable encoding with $L$ parameters (Theorem~\ref{thm:main}). The bound separates into the same margin-normalized term $(bE/\gamma)\sqrt{n/m}$ as in the quantum-data case and an additive encoding-capacity term of order $\sqrt{L/m}$, up to logarithmic factors. The proof introduces a two-scale covering argument that covers the encoding parameter space and the QNN unitary class at separate scales. A corollary for QNNs whose trainable parameters all lie in the encoding circuit (Corollary~\ref{cor:b0}) shows that the margin then enters the bound only logarithmically.
\item Following the predictor-evaluation methodology of \citet{jiang2019fantastic} and \citet{dziugaite2020search}, we conduct a systematic empirical study of 19{,}440 QNN models trained on binary classification tasks derived from MNIST, KMNIST, and Fashion-MNIST. The study sweeps four embedding circuits, four variational ans\"atze, three re-uploading depths, three qubit counts, and a grid of optimization hyperparameters. Five baseline complexity measures are evaluated with a conditional independence test and a sign-error analysis. Margin-based measures consistently outperform parameter-count measures on both criteria.
\item Guided by the derived bounds, we test two theorem-derived complexity measures, the composite $\sqrt{n_q}/m_m + c\sqrt{n_p}$ of Theorem~\ref{thm:main} and $n_p\ln(n_p/m_m)$ of Corollary~\ref{cor:b0}. Neither improves on the inverse margin. With $c$ tuned in-sample, the composite lowers the sign-error only slightly, within the bootstrap uncertainty. Both measures score below the inverse margin in the conditional independence test. We attribute this to the parameter count being fixed by the architecture hyperparameters, on which the test conditions.
\end{itemize}

Together, these results extend the margin-based analysis of \citet{hur2025margins} to QNNs trained on classical data and quantify the additional capacity introduced by learning the encoding. 

\section{Related Work}
\label{sec:related}

\paragraph{Margins and generalization in classical learning.}
Margin analysis dates back to support vector machines \citep{cortes1995support} and underlies modern norm-based generalization theory for deep networks, notably the spectrally normalized margin bounds of \citet{bartlett2017spectrally} and their PAC-Bayesian counterpart \citep{neyshabur2018pacbayesian}. Beyond bounds, large-scale empirical evaluations established protocols for judging whether a complexity measure genuinely \emph{predicts} generalization rather than merely correlating with it. \citet{jiang2019fantastic} introduced a conditional independence criterion inspired by probabilistic causality, and \citet{dziugaite2020search} advocated distributional robustness of sign-errors across experimental environments. We adopt both protocols.

\paragraph{Generalization in QML.}
Uniform generalization bounds for QML have been derived from parameter counts \citep{caro2022generalization}, effective dimension \citep{abbas2021power}, and quantum information-theoretic arguments \citep{banchi2021generalization}. \citet{gilfuster2024understanding} showed that QNNs overfit randomized labels, questioning the practical value of uniform bounds, in analogy with \citet{zhang2021understanding}. \citet{hur2025margins} established a margin-based bound for QNNs with quantum input data by covering the image of the data matrix under the unitary class, and demonstrated empirically that margin-based metrics dominate parameter-based ones. More recently, \citet{rodriguez2026pac} derived solution-dependent PAC-Bayesian bounds for layered quantum models. Our theory instead extends the quantum-data margin bound to classical inputs with trainable encodings, and our experiments evaluate margin-based predictors in that setting.

\paragraph{Data encoding.}
For classical data, the choice of quantum embedding determines the model's frequency content \citep{schuld2021effect} and its generalization behavior. \citet{caro2021encoding} derived generalization bounds for \emph{fixed} encodings via fat-shattering dimensions of the induced function class. Data re-uploading \citep{perezsalinas2020data} interleaves encoding and trainable layers. Trainable embeddings \citep{lloyd2020quantum} and neural quantum embeddings \citep{hur2024neural,liu2025neural,kim2026multi,choi2025optimizing} optimize the encoding to separate the classes in Hilbert space, which enlarges margins \citep{hur2025margins}. Our results cover both regimes: a margin-based bound that carries over to fixed encodings (Proposition~\ref{prop:fixed}), and a bound that quantifies the capacity cost of training the encoding (Theorem~\ref{thm:main}).

\section{Preliminaries}
\label{sec:prelim}

\paragraph{Multiclass classification with QNNs.}
We consider $k$-class classification of $n$-qubit quantum states, writing $N = 2^n$. A QNN applies a unitary $U$ from an accessible class $\UQ \subseteq \mathbb{U}(N)$, followed by a positive operator-valued measurement (POVM) $\{E_i\}_{i=1}^k$. For a pure input state $x \in \bbC^N$, the model outputs the vector of class probabilities
\begin{equation*}
h(x) = \big( x^\dagger U^\dagger E_i U x \big)_{i=1}^k \in [0,1]^k ,
\end{equation*}
and predicts $\arg\max_i h(x)_i$. Given $m$ i.i.d.\ samples $S = \{(x_i, y_i)\}_{i=1}^m$ from an unknown distribution $\cD$, the true risk is $R(h) = \mathbb{E}_{(x,y)\sim\cD} \ind(\arg\max_i h(x)_i \neq y)$.

\paragraph{Margins.}
The \emph{margin operator} $\cM(v, y) = v_y - \max_{i \neq y} v_i$ measures the gap between the correct class probability and the strongest competitor. The sample margin of $(x,y)$ is $\cM(h(x), y)$. For a margin level $\gamma > 0$, the ramp loss $\ell_\gamma : \bbR \to [0,1]$ is defined by $\ell_\gamma(t) = \min\{1, \max\{0, 1 - t/\gamma\}\}$. The empirical margin loss is $\widehat R_\gamma(h) = m^{-1} \sum_{i=1}^m \ind\big( h(x_i)_{y_i} \le \gamma + \max_{j \neq y_i} h(x_i)_j \big)$. Writing $\cF_\gamma$ for the class of functions $(x,y) \mapsto \ell_\gamma(\cM(h(x), y))$ induced by the hypothesis class, the standard margin bound \citep{mohri2018foundations} states that, for any $0<\delta<1$, with probability at least $1 - \delta$, for all $h$,
\begin{equation}
\label{eq:margin-lemma}
R(h) \le \widehat R_\gamma(h) + 2 \Rad\big( (\cF_\gamma)_{|S} \big) + 3 \sqrt{\tfrac{\ln(2/\delta)}{2m}},
\end{equation}
where $\Rad$ denotes the empirical Rademacher complexity of the loss class restricted to the sample.

\paragraph{Margin bound for quantum data.}
Two constants of the QNN control the Rademacher complexity in \eqref{eq:margin-lemma}. The \emph{measurement constant} is $E = \big( \sum_{i=1}^k \snorm{E_i}^2 \big)^{1/2}$, where $\snorm{\cdot}$ is the spectral norm. The \emph{unitary distance bound} $b$ satisfies $\| U - \Uref \|_{2,1} \le b$ for all $U \in \UQ$ and some reference unitary $\Uref$, where $\|\cdot\|_{2,1}$ is the sum of the Euclidean norms of the columns. Covering the image $\UQ X = \{ UX : U \in \UQ \}$ of the sample data matrix $X \in \bbC^{N \times m}$ with matrix sparsification techniques \citep{zhang2002covering}, \citet{hur2025margins} obtained the following.

\begin{theorem}[\citealp{hur2025margins}]
\label{thm:quantum}
Consider an $n$-qubit QNN with unitary class $\UQ$ (distance bound $b$) and POVMs $\{E_i\}_{i=1}^k$ for $k$-class classification of quantum states. For any $\gamma > 0$ and $0 < \delta < 1$, with probability at least $1 - \delta$ over an i.i.d.\ sample of size $m$, every hypothesis $h$ satisfies
\begin{equation*}
R(h) \le \widehat R_\gamma(h) + \Otil\left( \frac{bE}{\gamma} \sqrt{\frac{n}{m}} + \sqrt{\frac{\ln(2/\delta)}{m}} \right).
\end{equation*}
\end{theorem}

When the $k$ POVM elements are nonzero projectors, $E = \sqrt{k}$, so the measurement factor is common across models with the same number of classes. The remaining factors $b$ and $n$ must still be accounted for when comparing different unitary classes or circuit widths.

\section{Margin Bounds with Classical Data Encoding}
\label{sec:theory}

\paragraph{Fixed encodings.}
A fixed encoding transports the data distribution onto the space of quantum states, so the classical-data problem is an instance of the quantum-data one and the quantum-data bound transfers directly.

\begin{proposition}[Fixed encodings]
\label{prop:fixed}
For any fixed encoding $x \mapsto |\psi(x)\rangle$ of classical inputs into $n$-qubit states, the bound of Theorem~\ref{thm:quantum} holds for the hypothesis class of all $h(x) = \big( \psi(x)^\dagger U^\dagger E_i U \psi(x) \big)_{i=1}^k$ with $U \in \UQ$.
\end{proposition}

\begin{proof}
The encoded pairs $(\psi(x_i), y_i)$ are i.i.d.\ samples from the pushforward of $\cD$ under $(x,y) \mapsto (\psi(x), y)$, and under this identification the risks, the margins, and the hypothesis class are exactly those of the quantum-data setting.
\end{proof}

\paragraph{Trainable encodings.}
Modern encodings, however, are trained together with the classifier, and the reduction above then breaks down. The encoded data matrix is no longer fixed but moves with the encoding parameters. We model a trainable encoding as a parameterized circuit that prepares $|\psi(x;\phi)\rangle = V(x;\phi)\, |0\rangle^{\otimes n}$ with
\begin{equation}
\label{eq:encoding}
V(x;\phi) = \left( \prod_{j=1}^{L} e^{i H_j \phi_j}\, W_j(x) \right) W_0(x),
\end{equation}
where the $W_j(x)$ are fixed data-dependent unitaries (e.g., angle-embedding blocks), the $H_j$ are Hermitian generators, and $\phi = (\phi_1, \dots, \phi_L) \in [-\pi,\pi]^L$ contains the $L$ independent trainable encoding parameters. A circuit layer containing several trainable gates is decomposed into scalar-parameter factors of the form appearing in \eqref{eq:encoding}. The fixed unitary $W_0(x)$ allows the case $L=0$ to represent an arbitrary fixed encoding. Data re-uploading circuits \citep{perezsalinas2020data} with fixed input weights are obtained by repeating the same data-dependent block. Writing $g(z) = ( z^\dagger E_i z )_{i=1}^k$, define the hypothesis and induced loss classes
\begin{align}
\label{eq:class}
\mathcal{G}
&= \big\{ x \mapsto g(U\psi(x;\phi)) : U \in \UQ,\ \phi \in [-\pi,\pi]^L \big\}, \\
\cF_\gamma
&= \big\{ (x,y) \mapsto \ell_\gamma\big(\cM(h(x),y)\big) : h \in \mathcal{G} \big\}.
\end{align}

\begin{theorem}
\label{thm:main}
Consider an $n$-qubit QNN with unitary class $\UQ$ and distance bound $b>0$, POVMs $\{E_i\}_{i=1}^k$, and a trainable encoding of the form \eqref{eq:encoding} with $L$ parameters. For any $\gamma > 0$ and $0 < \delta < 1$, with probability at least $1 - \delta$ over an i.i.d.\ sample of size $m$, every hypothesis $h \in \mathcal{G}$ satisfies
\begin{equation*}
R(h) \le \widehat R_\gamma(h) + \Otil\left( \frac{bE}{\gamma} \sqrt{\frac{n}{m}} + \sqrt{\frac{L}{m}} + \sqrt{\frac{\ln(2/\delta)}{m}} \right).
\end{equation*}
\end{theorem}

The first term is the margin-normalized capacity of the post-encoding QNN unitary class, as in Theorem~\ref{thm:quantum}, while the second term accounts for the trainable encoder. Setting $L=0$ recovers the fixed-encoding case. The $\Otil$ notation suppresses logarithmic factors. Theorem~\ref{thm:full} in Appendix~A gives the explicit dependence.

Any variational classifier can be viewed as a trainable encoding followed by a fixed classifier, by absorbing all of its trainable gates into the encoding. This leads to the following corollary.

\begin{corollary}
\label{cor:b0}
In the setting of Theorem~\ref{thm:main} with $\UQ = \{I\}$ and $EH\sqrt{L} \ge 2\gamma$, with probability at least $1-\delta$ over an i.i.d.\ sample of size $m$, every hypothesis $h$ satisfies
\begin{equation*}
R(h) \le \widehat R_\gamma(h) + O\!\left( \sqrt{\frac{L}{m}\, \ln \frac{EH\sqrt{L}}{\gamma}} + \sqrt{\frac{\ln(2/\delta)}{m}} \right),
\end{equation*}
with absolute constants, where $H = (\sum_{j} \snorm{H_j}^2)^{1/2}$. If every trainable gate is a Pauli rotation and the measurement is projective, the encoding term is $O(\sqrt{(L/m)\ln(L/\gamma)})$.
\end{corollary}

The complexity term is a parameter count and the margin enters only through $\widehat R_\gamma$ and the logarithm, as for a finite-dimensional Lipschitz class whose fat-shattering dimension at scale $\gamma$ grows like $L\ln(1/\gamma)$~\citep{anthony2009neural}, in contrast to the polynomial $1/\gamma$ of Theorem~\ref{thm:main}. 
It is a margin-loss analogue of the parameter-count bound of~\citet{caro2022generalization}. Explicit constants are given in Appendix~\ref{app:corproof}.

\paragraph{Proof sketch.}
The full proof is in Appendix~A, and we outline the argument here. Write $X(\phi) = [\psi(x_1;\phi), \dots, \psi(x_m;\phi)] \in \bbC^{N \times m}$ for the encoded data matrix, whose columns are normalized states, so $\|X(\phi)\|_2 = \sqrt m$ for \emph{every} $\phi$. (1)~Since $\ell_\gamma(\cM(g(\cdot), y))$ is $(4E/\gamma)$-Lipschitz uniformly in $y$ \citep{hur2025margins}, covering numbers of $(\cF_\gamma)_{|S}$ reduce to covering numbers of the set $\{ U X(\phi) : U \in \UQ, \phi \in [-\pi,\pi]^L \}$, which is no longer the image of a single data matrix. (2)~The key step is a \emph{two-scale covering} argument that restores the single-matrix structure.

\begin{lemma}[Two-scale covering]
\label{lem:twoscale}
For any radii $r_1, r_2 > 0$,
\begin{equation*}
\cN\big( \{ U X(\phi) \},\, r_1 + r_2,\, \|\cdot\|_2 \big) \le
\cN\big( \{ X(\phi) \},\, r_1,\, \|\cdot\|_2 \big)
\cdot \sup_{\phi'} \cN\big( \{ U X(\phi') : U \in \UQ \},\, r_2,\, \|\cdot\|_2 \big).
\end{equation*}
\end{lemma}

The proof uses unitary invariance of the Frobenius norm. If $X(\phi)$ is close to a cover point $X(\phi')$, then $U X(\phi)$ is equally close to $U X(\phi')$, and the image $\UQ X(\phi')$ can be covered by the original argument. (3)~The unitary factor is bounded exactly as in the quantum-data case, since every $X(\phi')$ has Frobenius norm $\sqrt m$. (4)~For the encoding factor, a telescoping argument shows the map $\phi \mapsto X(\phi)$ is $(\sqrt m\, H)$-Lipschitz, so an $r$-net of the parameter cube $[-\pi,\pi]^L$ (of log-cardinality at most $L \ln(1 + \pi \sqrt{L} / r)$) induces a cover of $\{X(\phi)\}$. (5)~Dudley's entropy integral over the combined covering bound, with the discretization scale chosen proportional to $(bE/\gamma) \sqrt{\ln 4N^2}$, yields the stated complexity. Combining with \eqref{eq:margin-lemma} completes the proof.

The product form \eqref{eq:encoding} is sufficient but not necessary: the proof only requires the encoder family to be uniformly Lipschitz in its trainable parameters. Fully classical neural encoders require a separate covering bound and are not covered by Theorem~\ref{thm:main}. We next evaluate margin- and parameter-based statistics across QNNs with varying encoder capacity.

\section{Experimental Setup}
\label{sec:setup}

We follow the predictor-evaluation methodology developed for classical deep networks \citep{jiang2019fantastic,dziugaite2020search}. We train a large family of models spanning many hyperparameter axes, compute candidate complexity measures for each trained model, and score each measure by how well it predicts the generalization gap across controlled hyperparameter changes. Each hyperparameter is denoted in a bold font in the following subsections. 

\subsection{Datasets and Preprocessing}
We construct three binary image-classification tasks: MNIST \citep{lecun2010mnist} classes 1 vs.\ 6, KMNIST \citep{clanuwat2018deep} classes 0 vs.\ 5, and Fashion-MNIST \citep{xiao2017fashion} classes 4 vs.\ 9. The \textbf{dataset} itself is treated as a hyperparameter in the evaluation below. For each class, 5{,}000 training and 500 test samples are drawn at random. Images are standardized, reduced to two dimensions by principal component analysis, and standardized again so that the two features can be used directly as rotation angles of quantum gates. The two-dimensional feature vector is tiled to match the register width. For example, $[x_0, x_1]$ becomes $[x_0, x_1, x_0, x_1, x_0, x_1]$ on six qubits.

\subsection{Model Family}
Each QNN is specified by four architectural hyperparameters.

\paragraph{Embedding circuit $U_E(x)$.}
Four fixed embeddings are used: $R_X$ angle embedding, $R_Y$ angle embedding, and IQP embeddings with one or two layers \citep{havlicek2019supervised} (see Figure~\ref{fig:embedding}).

\paragraph{Variational circuit $V(\phi)$.}
Four ans\"atze are used: basic entangler layers with $R_X$ or $R_Y$ rotations, strongly entangling layers \citep{schuld2020circuit}, and simplified two-design layers \citep{cerezo2021cost} (see Figure~\ref{fig:vqc}).

\paragraph{Re-uploading depth and circuit width.}
We evaluate re-uploading depths $R \in \{4,6,8\}$ and circuit widths $n_q \in \{4,6,8\}$. The full circuit alternates embedding and variational blocks. The model prepares $\big( \prod_{i=0}^{R-1} V(\phi_i)\, U_E(x) \big) |0\rangle^{\otimes n}$ and outputs the measurement distribution of qubit~0. If the outermost variational block is viewed as the post-encoding QNN unitary, the preceding trainable rotations play the role of the $L$ encoding parameters in \eqref{eq:encoding}. 
The binary measurement on qubit~0 is projective, so $E=\sqrt{2}$ for every model.

\begin{figure*}[t]
    \centering
    \sbox0{
    \begin{adjustbox}{scale = 0.7}
        \begin{quantikz}
            \lstick{$|0\rangle$} & \gate{R_X(x_0)} & \\
            \lstick{$|0\rangle$} & \gate{R_X(x_1)} & \\
            \lstick{$|0\rangle$} & \gate{R_X(x_2)} & \\
            \lstick{$|0\rangle$} & \gate{R_X(x_3)} &
        \end{quantikz}
    \end{adjustbox}
    }
    \sbox1{
    \begin{adjustbox}{scale = 0.7}
        \begin{quantikz}
            \lstick{$|0\rangle$} & \gate{R_Y(x_0)} & \\
            \lstick{$|0\rangle$} & \gate{R_Y(x_1)} & \\
            \lstick{$|0\rangle$} & \gate{R_Y(x_2)} & \\
            \lstick{$|0\rangle$} & \gate{R_Y(x_3)} &
        \end{quantikz}
    \end{adjustbox}
    }
    \sbox2{
    \begin{adjustbox}{scale = 0.7}
        \begin{quantikz}
            \lstick{$|0\rangle$} & \gate{H} & \gate{R_Z(x_0)} & \gate[2]{R_{ZZ}(x_0x_1)} & & \\
            \lstick{$|0\rangle$} & \gate{H} & \gate{R_Z(x_1)} & & \gate[2]{R_{ZZ}(x_1x_2)} & \\
            \lstick{$|0\rangle$} & \gate{H} & \gate{R_Z(x_2)} & \gate[2]{R_{ZZ}(x_2x_3)} & & \\
            \lstick{$|0\rangle$} & \gate{H} & \gate{R_Z(x_3)} & & &
        \end{quantikz}
    \end{adjustbox}
    }
    \begin{tabular}{ccc}
         \usebox0 & \usebox1 & \usebox2\\
         (a) & (b) & (c)
    \end{tabular}
    \caption{Four-qubit embedding circuits with input $\vec{x} = [x_0, x_1, x_2, x_3]$. (a) $R_X$ angle embedding. (b) $R_Y$ angle embedding. (c) One IQP embedding layer. One- and two-layer variants are used.}
    \label{fig:embedding}
\end{figure*}
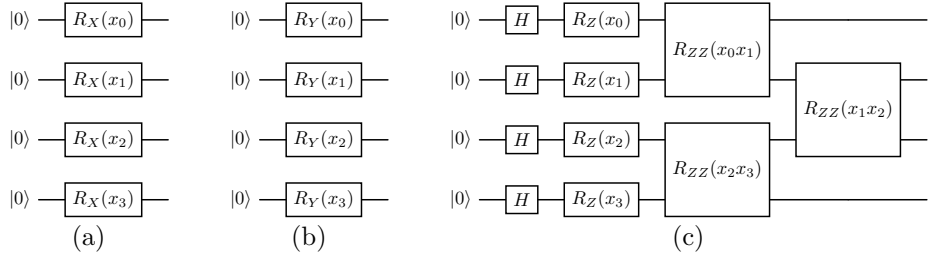

\begin{figure*}[t]
    \centering
    \sbox0{
    \begin{adjustbox}{scale = 0.7}
        \begin{quantikz}
            & \gate{R_X(\theta_0)} & \ctrl{1} & & & \targ{} & \\
            & \gate{R_X(\theta_1)} & \targ{} & \ctrl{1} & & & \\
            & \gate{R_X(\theta_2)} & & \targ{} & \ctrl{1} & & \\
            & \gate{R_X(\theta_3)} & & & \targ{} & \ctrl{-3} &
        \end{quantikz}
    \end{adjustbox}
    }
    \sbox1{
    \begin{adjustbox}{scale = 0.7}
        \begin{quantikz}
            & \gate{R_Y(\theta_0)} & \ctrl{1} & \gate{R_Y(\theta_4)} & & & \\
            & \gate{R_Y(\theta_1)} & \ctrl{0} & \gate{R_Y(\theta_5)} & \ctrl{1} & \gate{R_Y(\theta_8)} & \\
            & \gate{R_Y(\theta_2)} & \ctrl{1} & \gate{R_Y(\theta_6)} & \ctrl{0} & \gate{R_Y(\theta_9)} & \\
            & \gate{R_Y(\theta_3)} & \ctrl{0} & \gate{R_Y(\theta_7)} & & &
        \end{quantikz}
    \end{adjustbox}
    }
    \begin{tabular}{cc}
         \usebox0 & \usebox1 \\
         (a) & (b)
    \end{tabular}
    \caption{Four-qubit variational circuits with trainable parameters $\vec{\theta} = [\theta_0, \theta_1, \dots]$. (a) Basic entangler layer with $R_X$ rotations. The $R_Y$ variant substitutes $R_Y$ gates, and strongly entangling layers substitute general $U3$ rotations in place of $R_X$ rotations. (b) Simplified two-design layer.}
    \label{fig:vqc}
\end{figure*}
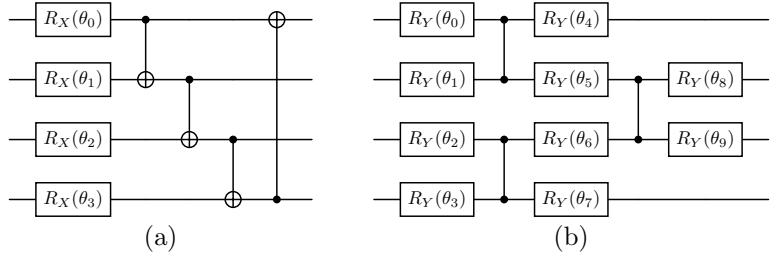

\subsection{Training}
Three optimization hyperparameters are varied, namely the \textbf{batch size} $\in \{32, 64, 128\}$, the \textbf{learning rate} $\in \{0.001, 0.002, 0.005\}$, and the \textbf{random seed} $\in \{16, 13, 42, 65, 89\}$. Models are trained using Adam \citep{kingma2015adam} until the following convergence criterion is met. Let $l_0$ and $\sigma_0$ denote the mean and standard deviation of the most recent 500 mini-batch losses, respectively, and let $l_1$ denote the mean of the preceding 500 mini-batch losses. Training stops when $|l_0 - l_1| < \sigma_0 / 1000$. In total, the grid comprises $4 \times 4 \times 3 \times 3$ architectures, $3 \times 3$ optimization settings, 3 datasets, and 5 seeds, i.e., $19{,}440$ training runs.

\subsection{Complexity Measures}
\label{sec:measures}
We compare five baseline complexity measures and two theorem-derived ones, all oriented so that \emph{larger} values should indicate a \emph{larger} generalization gap. Of the baselines, three contain margin information ($1/m_m$, $\sqrt{n_q}/m_m$, and $n_p/m_m$), while two are parameter-based ($n_p$ and $n_p \sqrt{n_q}$). Here $m_m$ is the median training margin, $n_q$ is the number of qubits, and $n_p$ is the total number of trainable parameters in the model. The factor $\sqrt{n_q}/m_m$ mirrors the explicit $\sqrt n/\gamma$ dependence in Theorem~\ref{thm:main}. The baselines isolate the ingredients of the bounds (margin, circuit width, and the parameter count that controls the encoding-capacity term), combining them only through simple products, so that comparisons among them attribute predictive power to the margin or to the parameter count.
The median is used rather than the lower quartile to avoid infinite inverse margins for inexpressive models \citep{hur2025margins}. Results with lower-quantile margins are given in Appendix~\ref{app:quantiles}.

The two theorem-derived measures combine the same ingredients in the functional forms that the bounds prescribe, and Section~\ref{sec:composite} compares them with the best baseline. The composite measure $\sqrt{n_q}/m_m + c \sqrt{n_p}$ matches the explicit dependence on $\gamma, n_q$ and $n_p$ in Theorem~\ref{thm:main}, with the margin and the encoding capacity as separate additive terms. The ratio $c$ stands in for the unknown unitary radius and is chosen to maximize the generalization property prediction accuracy of the measure. The measure $n_p\ln(n_p/m_m)$ follows the Pauli-rotation form of Corollary~\ref{cor:b0} with $L = n_p$ and $\gamma$ at the margin scale $m_m$, so that the margin enters only through a logarithm.

\subsection{Evaluation Criteria}
\label{sec:criteria}

\paragraph{Conditional independence test.}
The conditional independence test of \citet{jiang2019fantastic} measures the effect of a complexity measure $\mu$ on the generalization gap $g$ while controlling for confounding by the hyperparameters. For a pair of hyperparameter settings $(\theta_1, \theta_2)$, define the sign random variables $V_\mu = \mathrm{sign}(\mu(\theta_1) - \mu(\theta_2))$ and $V_g = \mathrm{sign}(g(\theta_1) - g(\theta_2))$, and let $U_S$ be uniformly distributed over the values of a subset $S$ of hyperparameter types. The criterion is the worst-case normalized conditional mutual information over small conditioning sets,
\begin{equation}
\label{eq:cit}
\cK(\mu) = \min_{U_S,\, |S| \le 2} \hat{\cI}(V_\mu, V_g \mid U_S),
\ 
\hat{\cI} = \frac{\cI(V_\mu, V_g \mid U_S)}{\cH(V_g \mid U_S)},
\end{equation}
where $\cI(\cdot,\cdot\mid\cdot)$ and $\cH(\cdot\mid\cdot)$ denote conditional mutual information and conditional entropy, estimated from the empirical distribution over runs. Larger $\cK(\mu)$ means $\mu$ carries more information about the generalization gap even after conditioning on the most adversarial pair of hyperparameters.

\paragraph{Sign-error distribution.}
Following \citet{dziugaite2020search}, the sign-error quantifies how often the ordering of two models by $\mu$ disagrees with their ordering by generalization gap, across pairs of settings $(\theta_1, \theta_2) \in \cE$ that differ in exactly one hyperparameter (random seeds excluded):
\begin{equation}
\label{eq:se}
\widehat{SE} = \tfrac{1}{2}\, \ind(n_{\mathrm{eff}} \ge 12) \sum_{(\theta_1, \theta_2) \in \cE} \bar\kappa(\theta_1, \theta_2)
\times \big[ 1 - \mathrm{sign}\big(g(\theta_1) - g(\theta_2)\big)\, \mathrm{sign}\big(\mu(\theta_1) - \mu(\theta_2)\big) \big].
\end{equation}
Each pair is weighted by $\bar\kappa(\theta_1, \theta_2)$, the normalization of $\kappa(\theta_1, \theta_2) = \max\big(0, \chi\big( (\hat g(\theta_1) - \hat g(\theta_2))/2, m_{\mathrm{te}} \big) - 0.5 \big)$. The factor $\chi(\epsilon, m) = (1-2\exp(-2m\epsilon^2))^2$ lower-bounds the probability that the sign of the gap difference, estimated from a test set of size $m$, agrees with its sign under the true distribution \citep{dziugaite2020search}. Thus, only hyperparameter-setting pairs with sign-consistency probability at least $0.5$ are retained. Some environments $\cE$ may consequently contain too few reliable comparisons, so a sign-error is calculated only when the effective sample size $n_{\mathrm{eff}} = (\sum_i \kappa_i)^2/\sum_i \kappa_i^2$ is at least 12. Some complexity measures have a large fraction of environments where the change in a hyperparameter does not result in a change of complexity measure, thus the sign-error is evaluated as 0.5. These measures may have an unfair advantage due to the large number of environments having relatively low sign-error. To resolve such issues, we remove all environments where the complexity measure remains constant.

Computing $\widehat{SE}$ within each experimental environment (each fixed configuration of the remaining hyperparameters) yields a \emph{distribution} of sign-errors per measure, whose upper tail reveals worst-case behavior. We report empirical cumulative distribution functions (ECDFs) and 90th-percentiles. 90th-percentiles indicate how well the complexity measure predicts the generalization gap for a large fraction of environments, thus indicating robustness as well as accuracy (\citet{dziugaite2020search}).
\section{Results}
\label{sec:results}

Training diagnostics and the distribution of empirical generalization gaps over all runs are reported in Appendix~\ref{app:training}.

\subsection{Baseline measures}
\label{sec:baseline}
\revtable{\begin{table}[t]\revcolor
    \centering
    \tightertab
    \begin{tabular}{lccccc}
        \toprule
        Complexity & & \multicolumn{2}{c}{$\widehat{SE}$, All} & \multicolumn{2}{c}{$\widehat{SE}$, Common} \\
        measure $\mu$ & $10^{3}\,\cK(\mu)$ & 90\% & Mean & 90\% & Mean \\
        \midrule
        $1/m_m$ & $\mathbf{4.565}$ & \textbf{0.7124} & \textbf{0.4799} & \textbf{0.7225} & \textbf{0.4783} \\
        $\sqrt{n_q}/m_m$ & 3.543 & 0.7144 & 0.4801 & 0.7274 & 0.4787 \\
        $n_p$ & 0.1932 & 0.7724 & 0.5038 & 0.7724 & 0.5038 \\
        $n_p/m_m$ & 0.4683 & 0.7256 & 0.4879 & 0.7531 & 0.4962 \\
        $n_p\sqrt{n_q}$ & 0.1564 & 0.7724 & 0.5038 & 0.7724 & 0.5038 \\
        \bottomrule
    \end{tabular}
    \caption{Five baseline complexity measures: conditional independence test score \eqref{eq:cit} (higher is better) and 90th-percentile and mean sign-error \eqref{eq:se} (lower is better) across all environments and across the Common environments in which the changed hyperparameter affects every measure (ansatz, number of re-uploadings, number of qubits). The maximum sign-error is 1 for every measure.}
    \label{tab:cit}
\end{table}
}

Table~\ref{tab:cit} reports both evaluation criteria for the five baseline complexity measures. In the conditional independence test, the pure margin measure $1/m_m$ attains the highest score. It exceeds the score of the qubit-normalized margin measure $\sqrt{n_q}/m_m$ and every parameter-based measure. The mixed measure $n_p/m_m$ underperforms compared to both inverse-margin measures, suggesting that multiplying margin information by a parameter count degrades rather than improves the signal.

The sign-error columns of Table~\ref{tab:cit} point the same way. No measure is completely robust (every measure exhibits environments with the maximal sign-error of 1), mirroring the classical findings of \citet{dziugaite2020search}. At the 90th-percentile, however, the margin-based measures achieve markedly lower sign-error than the parameter-based measures.

Figure~\ref{fig:hp_se_ecdf} decomposes the 90th-percentile sign-errors by the hyperparameter that is varied within each pair. Two observations stand out. First, the parameter-based measures $n_p$ and $n_p \sqrt{n_q}$ are \emph{constant} under changes of embedding, batch size, learning rate, and dataset, so those environments are excluded for them (Appendix~\ref{app:details}; dashes in Figure~\ref{fig:hp_se_ecdf}) and their scores rest on the three architecture axes, where the deficit is concentrated on the ansatz axis ($0.878$ versus $0.768$ for $1/m_m$). Second, the margin-based measures show lower 90th-percentile values for optimization hyperparameters (batch size and learning rate) compared to other hyperparameters, which is desirable: a good complexity measure should respond to whatever affects generalization, whichever hyperparameter caused it \citep{dziugaite2020search}.

\begin{figure}[t]
    \centering
    \includegraphics[width=0.55\columnwidth]{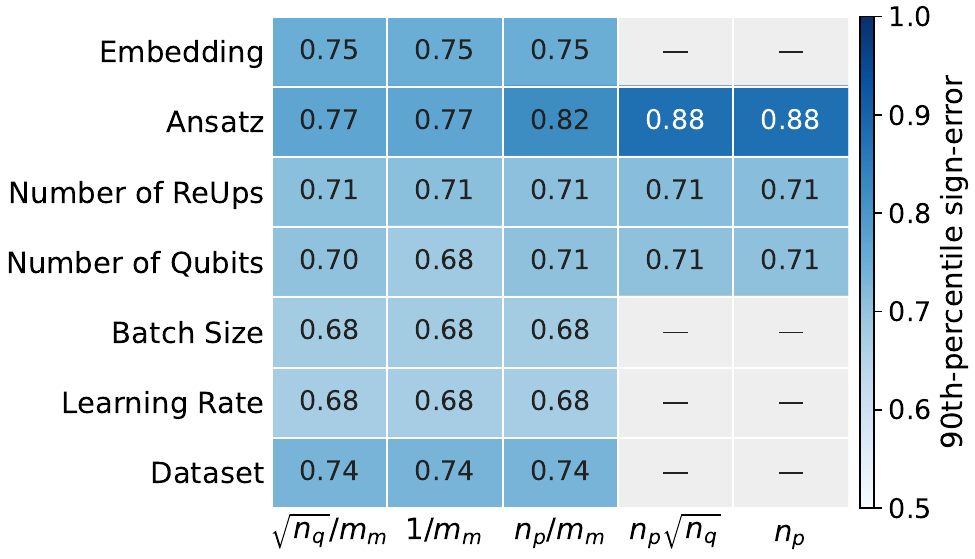}
    \caption{90th-percentile values of the sign-error empirical distribution, decomposed by the hyperparameter varied within each pair. Lower is better.}
    \label{fig:hp_se_ecdf}
\end{figure}

The overall advantage of margin-based measures could, in principle, stem solely from the axes on which parameter counts are constant. The Common columns of Table~\ref{tab:cit} therefore restrict the comparison to the hyperparameters that change \emph{every} measure, namely the ansatz, the number of re-uploadings, and the number of qubits. The differences shrink, but the margin-based measures still dominate at the 90th-percentile and mean. Full empirical cumulative distributions of sign-errors are further detailed in Appendix~\ref{app:se_ecdf} (Figure~\ref{fig:se_ecdf}).

\subsection{Theorem-derived measures}
\label{sec:composite}
Table~\ref{tab:derived} compares the best baseline, $1/m_m$, with the two theorem-derived measures of Section~\ref{sec:measures}, which combine the margin and the parameter count in the functional forms prescribed by the bounds. On the axes that leave $n_p$ and $n_q$ unchanged (embedding, batch size, learning rate, dataset) both order models exactly as $1/m_m$, so all differences arise on the architecture (Common) axes.

\revtable{\begin{table}[t]\revcolor
    \centering
    \tightertab
    \begin{tabular}{lccc}
        \toprule
        Complexity measure $\mu$ & $10^{3}\,\cK(\mu)$ & \multicolumn{2}{c}{90\% $\widehat{SE}$} \\
         & & All & Common \\
        \midrule
        $1/m_m$ & $\mathbf{4.565}$ & 0.7124 & 0.7225 \\
        $\sqrt{n_q}/m_m + c^*\sqrt{n_p}$ & 3.316 & $\mathbf{0.7117}$ & $\mathbf{0.7212}$ \\
        $n_p\ln(n_p/m_m)$ & 0.3657 & 0.7321 & 0.7699 \\
        \bottomrule
    \end{tabular}
    \caption{Best baseline versus the two theorem-derived measures of Section~\ref{sec:measures}: conditional independence score \eqref{eq:cit} (higher is better) and 90th-percentile sign-error over all and over the Common environments (lower is better); $c^*=0.075$. Confidence intervals are given in Appendix~\ref{app:ci}.}
    \label{tab:derived}
\end{table}
}

A composite complexity measure $\mu(c)$ of the form $\sqrt{n_q}/m_m + c\sqrt{n_p}$ that incorporates the number of qubits, number of parameters, and the margin following Theorem~\ref{thm:main} is optimized in-sample through grid search to minimize the 90th-percentile sign-error. The minimum 90th-percentile sign-error of $0.7117$ was achieved at $c^*=0.075$, an improvement of $0.0007$ over $1/m_m$ whose paired bootstrap interval includes zero ($[-0.0026, 0.0012]$; Appendix~\ref{app:ci}; held-out check in Appendix~\ref{app:composite}). When applied to the conditional independence test, the optimized composite complexity measure scored $\cK(\mu(c^*)) = 3.316\times10^{-3}$, which improves upon the parameter-based measures but not the pure margin. In fact, $\cK$ decreases monotonically in $c$ over the whole scanned range, from $3.55\times10^{-3}$ at $c=0$ to $6.6\times10^{-5}$ at $c=2$ (Table~\ref{tab:cscanK}), so the composite never improves on its base $\sqrt{n_q}/m_m$ in this test. 
The conditional independence test is designed to expose this. It scores the information a measure carries about the gap once the hyperparameters are fixed, and since $n_p$ is determined by three architecture hyperparameters, conditioning removes everything a term in $n_p$ could contribute. What remains is a pull of the model-pair ordering toward that of $n_p$, one of the weakest measures in Table~\ref{tab:cit}.

Since $c$ depends on the unitary radius $b$, it should depend on the ansatz type and qubit number. Optimizing one weight $c_{A,n_q}$ per (ansatz, $n_q$) pair lowers the 90th-percentile sign-error to $0.6988$ on the five axes that leave $c_{A,n_q}$ fixed (where $1/m_m$ scores $0.7052$), but only to $0.7108$ (versus $0.7124$, within the bootstrap uncertainty) over all seven axes, with $\cK(\mu(c^*_{A,n_q})) = 2.630\times 10^{-3}$, again below the pure margin measure (Appendix~\ref{app:composite}).

Corollary~\ref{cor:b0} places the parameter count outside and the margin inside a logarithm. The measure $n_p\ln(n_p/m_m)$ is markedly worse than $1/m_m$ on both criteria (Table~\ref{tab:derived}), because its ordering is dominated by the parameter count: $\ln n_p$ ranges from $2.77$ to $5.26$ across the grid, whereas $\ln(1/m_m)$ has median $0.74$, so the term $n_p\ln n_p$ exceeds the margin-dependent term $n_p\ln(1/m_m)$ in every run. Consequently $n_p\ln(n_p/m_m)$ orders $97.7\%$ of the pairs of runs with different $n_p$ exactly as $n_p$ does (Kendall $\tau = 0.92$ with $n_p$, $-0.06$ with $1/m_m$) and inherits the conditional independence score of the parameter count.

Finally, we score the empirical margin loss $\widehat R_\gamma$, i.e., the fraction of training points with margin at most $\gamma$. At each of $\gamma \in \{0.05, 0.1, 0.2\}$ it outperforms $1/m_m$ on both criteria, with $\cK$ up to $1.56\times10^{-2}$ and a 90th-percentile sign-error down to $0.706$, a gain whose paired bootstrap interval excludes zero (Appendix~\ref{app:rgamma}). The lower tail of the margin distribution thus carries more information about the gap than its median. Unlike $\widehat R_\gamma$, however, $1/m_m$ does not depend on a margin level, which makes it the simpler summary of the margin distribution.

\section{Conclusion}
\label{sec:conclusion}

We extended margin-based generalization theory for quantum neural networks from quantum to classical data. For a fixed data encoding, the quantum-data margin bound of \citet{hur2025margins} transfers without modification (Proposition~\ref{prop:fixed}). For trainable encodings, such as data re-uploading circuits, we proved a bound in which the margin-normalized term $(bE/\gamma)\sqrt{n/m}$ of the quantum-data case is accompanied by an additive encoding-capacity term of order $\sqrt{L/m}$ (Theorem~\ref{thm:main}), using a two-scale covering argument that treats the encoding parameters and the post-encoding unitary separately. When all trainable parameters lie in the encoding circuit, the complexity term reduces to a parameter count and the margin enters only logarithmically (Corollary~\ref{cor:b0}).

Empirically, across 19{,}440 QNNs trained on binary tasks from MNIST, KMNIST, and Fashion-MNIST, the inverse margin is the strongest of the measures compared in predicting the generalization gap, under both the conditional independence test and the sign-error analysis. It outperforms the parameter-count baselines as well as the measures derived from Theorem~\ref{thm:main} and Corollary~\ref{cor:b0}, which combine the margin with the parameter count. 
The composite measure improves on the sign-error of the inverse margin only when its weight is tuned in-sample, and then by less than the bootstrap uncertainty. Since the parameter count is determined by the architecture hyperparameters on which the test conditions, it adds little information about the generalization gap beyond what the margin already carries. These findings extend the quantum-data observations of \citet{hur2025margins} to classical data.

A natural follow-up question is how hardware noise affects these findings, as noise such as depolarizing noise contracts the measurement probabilities toward uniform and thereby shrinks the margin. On the theoretical side, Theorem~\ref{thm:main} assumes that trainable gates and data-dependent gates are separate, as in Equation~\eqref{eq:encoding}. Extending the analysis to encodings that process data and trainable parameters jointly, such as classical neural encoders \citep{hur2024neural} or rotation gates with trainable frequencies \citep{jaderberg2024let}, would bring a broader class of QNNs within the scope of the theory.

\subsection*{Code Availability}

The code can be found in the following repository: \url{https://github.com/jyryu98/margin-gen-qml-classical-encoding}.

\subsection*{Acknowledgments}
This work was supported by the Institute of Information \& Communications Technology Planning \& Evaluation (IITP) ITRC (Information Technology Research Center) funded by the the Ministry of Science and ICT (MSIT), Korea (IITP-2026-RS-2026-25519864), by the IITP grant (No. 2019-0-00003, Research and Development of Core Technologies for Programming, Running, Implementing and Validating of Fault-Tolerant Quantum Computing System), by the National Research Foundation of Korea (RS-2025-02309510), by Korean ARPA-H Project through the Korea Health Industry Development Institute (KHIDI), funded by the Ministry of Health \& Welfare, Korea (RS-2025-25456722), and National Institute of Health (NIH) research project (2026-ER0901-00).


\clearpage
\appendix
\onecolumn

\section{Proof of Theorem~\ref{thm:main}}
\label{app:proof}

Throughout, $n$ is the number of qubits, $N = 2^n$, $m$ the sample size, and $k$ the number of classes. Here $E = ( \sum_{i=1}^k \snorm{E_i}^2 )^{1/2}$ and $b>0$ satisfies $\| U - \Uref \|_{2,1} \le b$ for all $U \in \UQ$. We write $X(\phi) = [\psi(x_1;\phi), \dots, \psi(x_m;\phi)] \in \bbC^{N \times m}$ for the encoded sample matrix. Since its columns are normalized quantum states, $\| X(\phi) \|_2 = \sqrt m$ for every $\phi$, where $\|\cdot\|_2$ on matrices denotes the Frobenius norm. For a matrix $X$, we write $\UQ X = \{ UX : U \in \UQ \}$ for its image under the unitary class.

The argument first establishes the following explicit-logarithmic version of Theorem~\ref{thm:main} under two side conditions. Appendix~\ref{app:conditions} then removes those conditions, which only exclude regimes where Theorem~\ref{thm:main} holds for elementary reasons.

\begin{theorem}[Explicit-logarithmic version of Theorem~\ref{thm:main}]
\label{thm:full}
In the setting of Theorem~\ref{thm:main}, assume $bE \ge \gamma/4$ and $m \ge m_0 = 144\, b^2 E^2 \ln(4N^2) / \gamma^2$. Then the empirical Rademacher complexity in \eqref{eq:margin-lemma} obeys
\begin{equation*}
\Rad\big( (\cF_\gamma)_{|S} \big) \in O\Bigg(
\frac{bE}{\gamma} \sqrt{\frac{n}{m}} \left( 1 + \ln \frac{\gamma \sqrt m}{bE\sqrt n} \right)
+ \sqrt{ \frac{L}{m} \Big( 1 + \ln\Big( 1 + \frac{8\pi E H \sqrt L}{\gamma} \Big) \Big) }
\Bigg).
\end{equation*}
\end{theorem}

\subsection{Reduction to Matrix Covering}

The quantum-data analysis of \citet{hur2025margins} rests on two facts. First, for any label $y$, the map $z \mapsto \ell_\gamma(\cM(g(z), y))$ is $(4E/\gamma)$-Lipschitz in $\|\cdot\|_2$ on normalized states, since $\ell_\gamma$ is $(1/\gamma)$-Lipschitz, $\cM$ is $2$-Lipschitz in $\ell_\infty$, and each measurement function $g_i(z) = z^\dagger E_i z$ is $2\snorm{E_i}$-Lipschitz. Second, the image $\UQ X$ of a fixed data matrix admits the covering bound (their Equation~(9), obtained from Maurey sparsification \citep{zhang2002covering} applied to complex matrices)
\begin{equation}
\label{eq:orbit-cover}
\ln \cN\big( \{ U X : U \in \UQ \},\, r,\, \|\cdot\|_2 \big)
\le \left\lceil \frac{2 m b^2}{r^2} \right\rceil \ln 4N^2
\end{equation}
for any matrix $X$ with $\|X\|_2 = \sqrt m$ and any $r > 0$.

For the class \eqref{eq:class}, the Lipschitz property gives, for any $\epsilon > 0$,
\begin{equation}
\label{eq:lip-reduction}
\ln \cN\big( (\cF_\gamma)_{|S}, \epsilon, \|\cdot\|_2 \big)
\le \ln \cN\Big( \{ U X(\phi) : U \in \UQ,\ \phi \in [-\pi,\pi]^L \},\, \tfrac{\epsilon\gamma}{4E},\, \|\cdot\|_2 \Big),
\end{equation}
but the right-hand side is no longer the image of a single matrix. The data matrix moves with $\phi$.

\subsection{Two-Scale Covering}

\begin{lemma}[Lemma~\ref{lem:twoscale}, restated]
For any $r_1, r_2 > 0$,
\begin{equation*}
\cN\big( \{ U X(\phi) \},\, r_1 + r_2,\, \|\cdot\|_2 \big)
\le \cN\big( \{ X(\phi) : \phi \in [-\pi,\pi]^L \},\, r_1,\, \|\cdot\|_2 \big)
\cdot \sup_{\phi'} \cN\big( \{ U X(\phi') : U \in \UQ \},\, r_2,\, \|\cdot\|_2 \big).
\end{equation*}
\end{lemma}

Here the first covering number is over covers by points of $\{X(\phi)\}$ itself (internal covers): the grid of Lemma~\ref{lem:cube} is one, and an internal covering number is at most the external one at half the radius.

\begin{proof}
Take $C_\Phi \subset [-\pi,\pi]^L$ such that $\{ X(\phi') : \phi' \in C_\Phi \}$ is an $r_1$-cover of $\{ X(\phi) \}$. The grid construction of Lemma~\ref{lem:cube} below produces the cover from points \emph{inside} the cube, so every cover point is itself an admissible data matrix. For each $\phi' \in C_\Phi$, let $C_{\phi'}$ be an $r_2$-cover of the image $\UQ X(\phi')$. Given any $(U, \phi)$, choose $\phi' \in C_\Phi$ with $\| X(\phi) - X(\phi') \|_2 \le r_1$. By unitary invariance of the Frobenius norm, $\| U X(\phi) - U X(\phi') \|_2 = \| X(\phi) - X(\phi') \|_2 \le r_1$. Then choose $Z \in C_{\phi'}$ with $\| U X(\phi') - Z \|_2 \le r_2$. Hence $\bigcup_{\phi' \in C_\Phi} C_{\phi'}$ is an $(r_1 + r_2)$-cover of cardinality at most the claimed product.
\end{proof}

Splitting the radius in \eqref{eq:lip-reduction} as $\frac{\epsilon\gamma}{4E} = \frac{\epsilon\gamma}{8E} + \frac{\epsilon\gamma}{8E}$ and applying the lemma,
\begin{equation}
\label{eq:split}
\ln \cN\big( (\cF_\gamma)_{|S}, \epsilon, \|\cdot\|_2 \big)
\le \ln \cN\Big( \{ X(\phi) \},\, \tfrac{\epsilon\gamma}{8E},\, \|\cdot\|_2 \Big)
+ \sup_{\phi'} \ln \cN\Big( \{ U X(\phi') : U \in \UQ \},\, \tfrac{\epsilon\gamma}{8E},\, \|\cdot\|_2 \Big).
\end{equation}
Since every $X(\phi')$ satisfies $\| X(\phi') \|_2 = \sqrt m$, the second term follows from \eqref{eq:orbit-cover} at radius $\epsilon\gamma/8E$:
\begin{equation}
\label{eq:orbit-term}
\sup_{\phi'} \ln \cN\Big( \{ U X(\phi') \},\, \tfrac{\epsilon\gamma}{8E},\, \|\cdot\|_2 \Big)
\le \left\lceil \frac{128\, m b^2 E^2}{\epsilon^2 \gamma^2} \right\rceil \ln 4N^2 .
\end{equation}

\subsection{Lipschitz Property of the Encoding}

\begin{lemma}
\label{lem:enc-lip}
For all $\phi, \phi' \in [-\pi,\pi]^L$,
\begin{equation*}
\| \psi(x; \phi) - \psi(x; \phi') \|_2 \le H \| \phi - \phi' \|_2,
\end{equation*}
with $H = \big( \sum_{j=1}^L \snorm{H_j}^2 \big)^{1/2}$, and consequently $\| X(\phi) - X(\phi') \|_2 \le \sqrt m\, H\, \| \phi - \phi' \|_2$.
\end{lemma}

\begin{proof}
For unitaries $A_1, \dots, A_L$ and $B_1, \dots, B_L$, a telescoping argument gives $\snorm{ \prod_j A_j - \prod_j B_j } \le \sum_j \snorm{ A_j - B_j }$, since all remaining factors have unit spectral norm. The fixed interleaved unitaries $W_j(x)$ cancel in each telescoping step. For a single layer,
\begin{equation*}
\snorm{ e^{iHa} - e^{iHb} }
= \Big\lVert iH \int_b^a e^{iHt}\, dt \Big\rVert_\sigma
\le \snorm{H}\, |a - b| .
\end{equation*}
Combining the two facts and applying the Cauchy--Schwarz inequality,
\begin{equation*}
\| \psi(x;\phi) - \psi(x;\phi') \|_2
\le \sum_{j=1}^L \snorm{H_j} |\phi_j - \phi_j'|
\le H \| \phi - \phi' \|_2 .
\end{equation*}
The matrix statement follows because each of the $m$ columns changes by at most $H \|\phi - \phi'\|_2$ in Euclidean norm.
\end{proof}

\paragraph{More general parameterized layers.}
The subsequent argument uses only the uniform Lipschitz property in Lemma~\ref{lem:enc-lip}. It therefore also applies to encoders $V(x;\phi)=\prod_{j=1}^L V_j(x;\phi_j)$ satisfying $\snorm{\partial_{\phi_j}V_j(x;\phi_j)}\le c_j$ uniformly in $x$ and $\phi_j$, with $H$ replaced by $(\sum_j c_j^2)^{1/2}$. This includes data--parameter coupled rotations such as $e^{i\phi_j x_\ell H}$ when the data are bounded. Encoders implemented by a separate classical neural network require a covering bound for that network and are outside Theorem~\ref{thm:main}.

\subsection{Covering the Parameter Cube}

\begin{lemma}
\label{lem:cube}
For any $r > 0$,
\begin{equation*}
\ln \cN\big( [-\pi,\pi]^L,\, r,\, \|\cdot\|_2 \big)
\le L \ln\Big( 1 + \frac{\pi \sqrt L}{r} \Big).
\end{equation*}
\end{lemma}

\begin{proof}
An axis-aligned grid of spacing $2r/\sqrt L$ inside $[-\pi,\pi]^L$ matches each coordinate to within $r/\sqrt L$ and is therefore an $r$-cover in $\|\cdot\|_2$, with $\lceil \pi \sqrt L / r \rceil$ points per axis, giving $\ln \cN \le L \ln \lceil \pi \sqrt L / r \rceil \le L \ln( 1 + \pi \sqrt L / r )$.
\end{proof}

We state the bound in the $\ln(1 + \cdot)$ form, standard for covering numbers of bounded sets (compare $\cN(B_2^d, r, \|\cdot\|_2) \le (1 + 2R/r)^d$ for the radius-$R$ ball). The covering number of a nonempty set is at least one, so a valid upper bound on $\ln\cN$ must remain nonnegative for every $r > 0$. An expression of the form $L \ln(K/\epsilon)$ fails this for $\epsilon > K$, while $\ln(1+\cdot)$ needs no side condition. In the regime $\epsilon \ll K$ relevant inside Dudley's integral, the two forms agree up to constants.

By Lemma~\ref{lem:enc-lip}, an $\frac{\epsilon\gamma}{8EH\sqrt m}$-cover of the cube induces an $\frac{\epsilon\gamma}{8E}$-cover of $\{ X(\phi) \}$, so
\begin{equation}
\label{eq:enc-term}
\ln \cN\Big( \{ X(\phi) \},\, \tfrac{\epsilon\gamma}{8E},\, \|\cdot\|_2 \Big)
\le L \ln\Big( 1 + \frac{K}{\epsilon} \Big),
\end{equation}
where $K = 8 \pi E H \sqrt{mL} / \gamma$.

\subsection{Covering Bound for the Hypothesis Class}

Under the assumption $bE \ge \gamma/4$ of Theorem~\ref{thm:full}, for all $0 < \epsilon \le \sqrt m$ one has $128\, m b^2 E^2 / (\epsilon^2 \gamma^2) \ge 8$, so $\lceil x \rceil \le \tfrac 98 x$ (valid for $x \ge 8$) absorbs the ceiling in \eqref{eq:orbit-term}. Radii $\epsilon > \sqrt m$ are irrelevant since $(\cF_\gamma)_{|S} \subseteq [0,1]^m$ has diameter at most $\sqrt m$. Combining \eqref{eq:split}, \eqref{eq:orbit-term}, and \eqref{eq:enc-term} gives, for all $0 < \epsilon \le \sqrt m$,
\begin{equation}
\label{eq:final-cover}
\ln \cN\big( (\cF_\gamma)_{|S}, \epsilon, \|\cdot\|_2 \big)
\le L \ln\Big( 1 + \frac{K}{\epsilon} \Big) + \frac{A}{\epsilon^2},
\end{equation}
with $A = 144\, m b^2 E^2 \ln(4N^2) / \gamma^2$.

\subsection{Dudley's Entropy Integral}

Dudley's entropy integral \citep[Lemma~A.5]{bartlett2017spectrally} gives
\begin{equation*}
\Rad\big( (\cF_\gamma)_{|S} \big)
\le \inf_{\alpha > 0} \Bigg( \frac{4\alpha}{\sqrt m}
+ \frac{12}{m} \int_\alpha^{\sqrt m} \sqrt{ \ln \cN\big( (\cF_\gamma)_{|S}, \beta, \|\cdot\|_2 \big) }\, d\beta \Bigg).
\end{equation*}
The lemma is stated for sets containing the origin, and applying it to $(\cF_\gamma)_{|S} \cup \{0\}$ increases each covering number by at most one element, affecting only absolute constants. Using $\sqrt{a + b} \le \sqrt a + \sqrt b$ on \eqref{eq:final-cover}, the two contributions are
\begin{equation*}
\frac{12}{m} \int_\alpha^{\sqrt m} \frac{\sqrt A}{\beta}\, d\beta
= \frac{12 \sqrt A}{m} \ln \frac{\sqrt m}{\alpha}
\qquad\text{and}\qquad
\frac{12}{m} \int_\alpha^{\sqrt m} \sqrt{ L \ln(1 + K/\beta) }\, d\beta
\le 12 \sqrt{\frac{L}{m}} \left( \sqrt{ \ln(1 + K') } + 0.89 \right),
\end{equation*}
where $K' = K/\sqrt m = 8\pi E H \sqrt L/\gamma$. The second bound follows by extending the integral down to $0$, which is allowed because its integrand is nonnegative and integrable at $0$: substituting $\beta = \sqrt m\, u$ and using $\ln(1 + K'/u) \le \ln(1 + K') + \ln(1/u)$ for $0 < u \le 1$ together with $\int_0^1 \sqrt{\ln(1/u)}\, du = \sqrt\pi/2 \le 0.89$, exactly as in Appendix~\ref{app:corproof}. The cutoff $\alpha$ is thus needed only for the unitary term. Choose
\begin{equation*}
\alpha = \sqrt{\frac{A}{m}} = \frac{12\, bE \sqrt{\ln 4N^2}}{\gamma},
\end{equation*}
which is admissible ($\alpha \le \sqrt m$) precisely when $m \ge m_0$, the sample-size condition of Theorem~\ref{thm:full}. Then
\begin{equation}
\label{eq:rad-final}
\Rad\big( (\cF_\gamma)_{|S} \big)
\le \frac{4 \sqrt A}{m} + \frac{12 \sqrt A}{m} \ln \frac{m}{\sqrt A} + 12 \sqrt{\frac{L}{m}} \left( \sqrt{ \ln(1 + K') } + 0.89 \right),
\end{equation}
with
\begin{equation*}
\frac{\sqrt A}{m} = \frac{12\, bE}{\gamma} \sqrt{\frac{\ln 4N^2}{m}} .
\end{equation*}
Since $\ln 4N^2 = 2(n+1)\ln 2 = \Theta(n)$ and $\sqrt a + 0.89 \le \sqrt{2(1 + a)}$ for $a \ge 0$,
\begin{align*}
\Rad\big( (\cF_\gamma)_{|S} \big)
&\in O\Bigg( \frac{bE}{\gamma} \sqrt{\frac nm} \left( 1 + \ln \frac{\gamma \sqrt m}{bE \sqrt n} \right)
+ \sqrt{ \frac Lm \Big( 1 + \ln\Big( 1 + \frac{8\pi E H \sqrt L}{\gamma} \Big) \Big) } \Bigg) \\
&\in \Otil\left( \frac{bE}{\gamma} \sqrt{\frac nm} + \sqrt{\frac Lm} \right).
\end{align*}
This establishes Theorem~\ref{thm:full}, and substituting into the margin bound \eqref{eq:margin-lemma} yields Theorem~\ref{thm:main} whenever $bE \ge \gamma/4$ and $m \ge m_0$. \hfill$\qed$

\subsection{Removing the Side Conditions}
\label{app:conditions}

It remains to prove Theorem~\ref{thm:main} when either side condition of Theorem~\ref{thm:full} fails.

\paragraph{Small samples.}
Suppose $m < m_0$. Then $(bE/\gamma) \sqrt{\ln(4N^2)/m} > 1/12$, and since $\ln 4N^2 = 2(n+1) \ln 2 \le 4 n \ln 2$ for $n \ge 1$, the first term of Theorem~\ref{thm:main} satisfies $(bE/\gamma) \sqrt{n/m} > 1/20$. Once the constant hidden in the $\Otil$ is at least $20$, the right-hand side of the bound exceeds one, and the bound holds trivially because $R(h) \le 1$.

\paragraph{Weak unitary classes.}
Suppose $0<bE < \gamma/4$ and $m \ge m_0$. The ceiling in \eqref{eq:orbit-term} can then no longer be absorbed multiplicatively, so we keep it via $\lceil x \rceil \le x + 1$, and \eqref{eq:final-cover} acquires the additive term $\ln 4N^2$. This term is active only at small radii. Since $\snorm{U - \Uref} \le \| U - \Uref \|_2 \le \| U - \Uref \|_{2,1} \le b$, every matrix in the image $\UQ X(\phi')$ satisfies $\| U X(\phi') - \Uref X(\phi') \|_2 \le \snorm{U - \Uref}\, \| X(\phi') \|_2 \le b \sqrt m$, so the single point $\Uref X(\phi')$ covers $\UQ X(\phi')$ at radius $b \sqrt m$. The unitary factor in \eqref{eq:split} is therefore zero at all radii $\epsilon \gamma / (8E) \ge b \sqrt m$, that is for all $\epsilon \ge D = 8\, b E \sqrt m / \gamma$. The additive term hence contributes to Dudley's entropy integral at most
\begin{equation*}
\frac{12}{m} \int_0^{D} \sqrt{\ln 4N^2}\, d\beta
= \frac{96\, bE}{\gamma} \sqrt{\frac{\ln 4N^2}{m}}
\in O\left( \frac{bE}{\gamma} \sqrt{\frac{n}{m}} \right),
\end{equation*}
which is absorbed into the first term of Theorem~\ref{thm:main}. The remainder of the proof proceeds unchanged, completing the proof for every $b>0$. When $b=0$, the post-encoding unitary class is a singleton; applying Dudley's integral directly to the encoder cover \eqref{eq:enc-term} gives the simpler encoder-only scaling $\Otil(\sqrt{L/m})$. \hfill$\qed$

\subsection{Proof of Corollary~\ref{cor:b0}}
\label{app:corproof}
With $\UQ = \{I\}$, the reduction \eqref{eq:lip-reduction} reads $\ln\cN((\cF_\gamma)_{|S}, \epsilon, \|\cdot\|_2) \le \ln\cN(\{X(\phi) : \phi \in [-\pi,\pi]^L\}, \epsilon\gamma/4E, \|\cdot\|_2)$, with no unitary factor. By Lemma~\ref{lem:enc-lip} an $r$-cover of the cube with $r = \epsilon\gamma/(4EH\sqrt m)$ induces an $\epsilon\gamma/4E$-cover of $\{X(\phi)\}$, so Lemma~\ref{lem:cube} gives
\begin{equation*}
\ln\cN\big((\cF_\gamma)_{|S}, \epsilon, \|\cdot\|_2\big) \le L \ln\Big(1 + \frac{K_0}{\epsilon}\Big), \qquad K_0 = \frac{4\pi E H \sqrt{mL}}{\gamma}.
\end{equation*}
Since $\beta \mapsto \sqrt{\ln(1 + K_0/\beta)}$ is integrable at $0$, Dudley's integral may be taken with lower limit $\alpha \to 0$, which removes the $4\alpha/\sqrt m$ term. Applying it to $(\cF_\gamma)_{|S} \cup \{0\}$ and substituting $\beta = \sqrt m\, u$, with $K_0' = K_0/\sqrt m = 4\pi E H \sqrt L/\gamma$ and the elementary bound $\ln(1 + K_0'/u) \le \ln(1 + K_0') + \ln(1/u)$ for $0 < u \le 1$,
\begin{align*}
\Rad\big((\cF_\gamma)_{|S}\big)
&\le \frac{12}{\sqrt m} \int_0^1 \Big( \sqrt{L \ln(1 + K_0'/u)} + \sqrt{\ln 2} \Big)\, du \\
&\le 12 \sqrt{\frac{L}{m}} \Big( \sqrt{\ln(1 + K_0')} + \frac{\sqrt\pi}{2} \Big) + \frac{12\sqrt{\ln 2}}{\sqrt m},
\end{align*}
using $\int_0^1 \sqrt{\ln(1/u)}\, du = \sqrt\pi/2 \le 0.89$. Substituting into the margin bound \eqref{eq:margin-lemma} gives, with probability at least $1-\delta$,
\begin{equation}
\label{eq:corexplicit}
R(h) \le \widehat R_\gamma(h) + 24 \sqrt{\frac{L}{m}} \left( \sqrt{\ln\!\Big(1 + \frac{4\pi E H \sqrt L}{\gamma}\Big)} + 0.89 \right) + \frac{20}{\sqrt m} + 3\sqrt{\frac{\ln(2/\delta)}{2m}}.
\end{equation}
When $E H \sqrt L \ge 2\gamma$ the logarithm is at least $\ln(1 + 8\pi) > 3$, so the bracket is at most a constant multiple of $\sqrt{\ln(1 + K_0')}$ and $\ln(1 + 4\pi x) \le 5\ln x$ for $x \ge 2$; this yields Corollary~\ref{cor:b0} with absolute constants. If every trainable gate is a Pauli rotation $\exp(-i\theta\sigma/2)$, then $\snorm{H_j} = 1/2$, $H = \sqrt L/2$, and $K_0' = 2\pi E L/\gamma$. If, in addition, the measurement is projective with $k$ outcomes, each $E_i$ is an orthogonal projector, so $\snorm{E_i} \le 1$, $E \le \sqrt k$, and $K_0' \le 2\pi\sqrt k\, L/\gamma$. We may take $\gamma < 1$, since margins never exceed $1$ and $\widehat R_\gamma(h) = 1$ otherwise; then $L/\gamma \ge 2$ for $L \ge 2$, and $\ln(1 + K_0') \le \ln(1 + 2\pi\sqrt k) + \ln(L/\gamma) \le c_k \ln(L/\gamma)$ with $c_k = 1 + \ln(1+2\pi\sqrt k)/\ln 2$ ($c_2 < 4.4$ for binary classification). The same lower bound $\ln(L/\gamma) \ge \ln 2$ absorbs the additive constant $0.89$ in \eqref{eq:corexplicit}, giving the stated $O(\sqrt{(L/m)\ln(L/\gamma)})$ form, with a constant depending only on $k$. \hfill$\qed$

\section{Training Diagnostics}
\label{app:training}

Figure~\ref{fig:train_diagnostics} summarizes the training results in all runs. Most runs achieve high training and test accuracy, with a population of poorly performing outliers concentrated among the smallest architectures. For two of the three datasets, the empirical generalization gap (training minus test accuracy) is distributed approximately symmetrically around zero rather than being predominantly positive. The gaps in this regime are small, so fluctuations from the finite test set can be comparable to the gaps themselves. The $\kappa$-weighting in Equation~\eqref{eq:se} discounts comparisons for which the ordering is not reliably determined.

\begin{figure}[h]
    \centering
    \includegraphics[width=0.8\columnwidth]{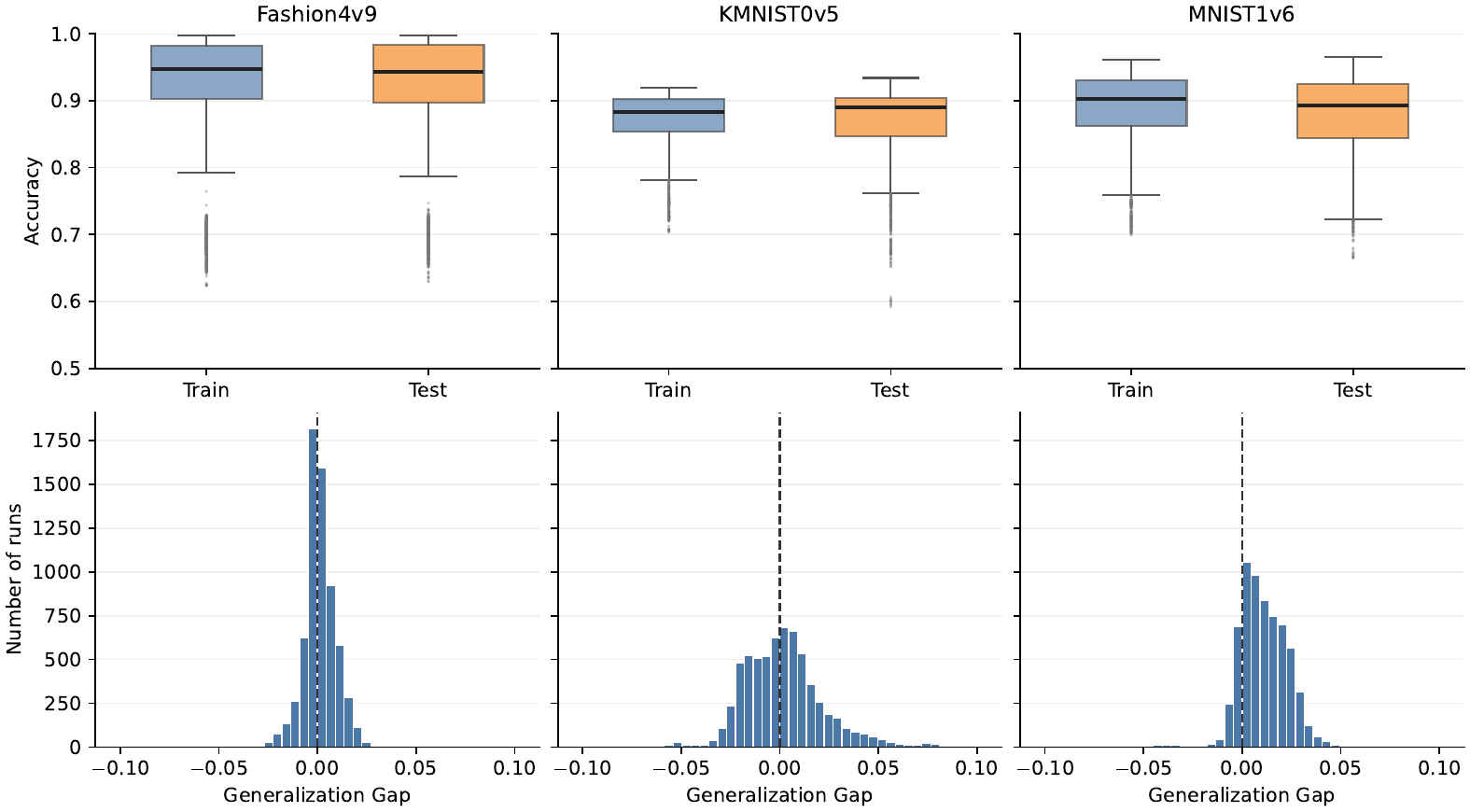}
    \caption{Accuracies and generalization gaps of all 19{,}440 runs. The training accuracy and test accuracies are represented as boxplots. The box extends from the first quartile to the third, with a line at the median. Error bars represent 1.5$\times$ the inter-quartile range added to the box boundaries. }
    \label{fig:train_diagnostics}
\end{figure}

\section{Additional Experimental Details}
\label{app:details}

\paragraph{Evaluation details.}
For the conditional independence test defined in Equation~\eqref{eq:cit}, the empirical joint distribution over runs is used to estimate the conditional mutual information and conditional entropy, following \citet{hur2025margins}. For the sign-error defined in Equation~\eqref{eq:se}, the function $\chi(\cdot, m_{\mathrm{te}})$ is the sign-consistency probability bound of \citet{dziugaite2020search} evaluated with the test-set size $m_{\mathrm{te}} = 1000$. Pairs are formed only between settings differing in exactly one hyperparameter, seeds are aggregated within each setting, and environments with effective sample size $n_{\mathrm{eff}} < 12$ are discarded. Environments where the complexity measure is constant about the change in the hyperparameter are also discarded. In Equation~\eqref{eq:cit}, $S$ ranges over all eight hyperparameter types of Section~\ref{sec:setup}, including the random seed, i.e.\ 36 conditioning sets (8 single types and 28 pairs); for every measure reported, the minimizing set does not involve the seed, so the 28 seed-free sets give the same $\cK$. Of the $30{,}732$ environments that contain at least one retained pair ($\kappa > 0$), $24{,}909$ satisfy $n_{\mathrm{eff}} \ge 12$ (median $n_{\mathrm{eff}} = 20.4$ among them, of at most 25). The margin-based measures are scored on these $24{,}909$ environments ($11{,}104$ Common), whereas $n_p$ and $n_p\sqrt{n_q}$ are scored on $10{,}443$ in both columns, because $n_p$ is also constant between the two basic-entangler ans\"atze, which removes $661$ further environments. Both criteria depend only on how a measure orders pairs of models, so strictly increasing transformations of a measure, such as the square root in Corollary~\ref{cor:b0} or in $\sqrt{n_p}$, leave all scores unchanged.

\paragraph{Margin computation.}
Margins are computed on the training set after convergence as $\cM(h(x), y) = h(x)_y - \max_{j \neq y} h(x)_j$, where $h(x)$ is the measurement distribution of qubit~0, and $m_m$ denotes the median of this distribution.

\paragraph{Software and hardware.}
Simulations were performed with the \texttt{default.qubit} statevector simulator in PennyLane 0.42.3. PyTorch 2.8.0 was the main machine learning framework used. Experiments were carried out on a desktop computer equipped with an AMD Ryzen 7 7800X3D 8-core processor and 32GBs of RAM.

\section{Composite Measure Optimization}
\label{app:composite}

We optimize $\mu(c) = \sqrt{n_q}/m_m + c\sqrt{n_p}$ to minimize the 90th-percentile sign-error by grid-searching $c$ over $[0.01, 0.1]$ with increments of $0.005$. The sign-errors formed a U-shape across this sweep as shown in Figure~\ref{fig:cscan}(a), achieving a minimum 90th-percentile sign-error of $0.7117$ at $c^*=0.075$.

Table~\ref{tab:cscanK} reports both criteria along a wider sweep of $c$: the 90th-percentile sign-error is U-shaped with its minimum at $c^*=0.075$, whereas $\cK$ decreases monotonically over the whole sweep, from its $c=0$ value $\cK(\sqrt{n_q}/m_m)$ to $6.6\times10^{-5}$ at $c=2$, below even the parameter-count ordering recovered as $c\to\infty$, so no choice of $c$ improves the conditional independence score. 

\begin{table}[h]
\centering
\tightertab
\begin{tabular}{lccccccccccccc}
\toprule
$c$ & 0 & 0.01 & 0.025 & 0.05 & $\mathbf{0.075}$ & 0.1 & 0.15 & 0.2 & 0.3 & 0.5 & 1 & 2 & $\infty$ \\
\midrule
$10^{3}\,\cK(\mu(c))$ & 3.546 & 3.518 & 3.485 & 3.407 & 3.316 & 3.269 & 3.139 & 3.006 & 2.731 & 2.284 & 1.130 & 0.066 & 0.194 \\
$\widehat{SE}_{90}$, All & 0.7144 & 0.7136 & 0.7134 & 0.7121 & \textbf{0.7117} & 0.7127 & 0.7129 & 0.7138 & 0.7154 & 0.7183 & 0.7233 & 0.7294 & 0.7724 \\
$\widehat{SE}_{90}$, Com. & 0.7274 & 0.7250 & 0.7247 & 0.7221 & \textbf{0.7212} & 0.7233 & 0.7237 & 0.7257 & 0.7310 & 0.7385 & 0.7486 & 0.7592 & 0.7724 \\
\bottomrule
\end{tabular}
\caption{Composite measure $\mu(c) = \sqrt{n_q}/m_m + c\sqrt{n_p}$ along a sweep of $c$: conditional independence score (higher is better) and 90th-percentile sign-error over all and over the Common environments (lower is better). $c=0$ is $\sqrt{n_q}/m_m$; the column $c=\infty$ is the ordering of $\sqrt{n_p}$, i.e., of $n_p$; its sign-errors are over the $10{,}443$ environments in which $n_p$ varies (finite $c$: all $24{,}909$). For reference, $\cK(1/m_m) = 4.570\times10^{-3}$ with sign-errors $0.7124$ and $0.7225$.}
\label{tab:cscanK}
\end{table}

We also optimize $\mu(c_{A, n_q}) = \sqrt{n_q}/m_m + c_{A,n_q}\sqrt{n_p}$, where each (qubit number, ansatz) pair has its own independent $c_{A, n_q}$. Each setting had its $c$ value grid-searched across $[0, 0.2]$ with increments of $0.005$. The 90th-percentile sign-errors for each setting pair for this sweep is shown in Figure~\ref{fig:cscan}(b). Using the optimum $c_{A, n_q}$ values, the aggregate sign-error 90th-percentile is $0.6988$. This aggregate runs over the environments in which $c_{A,n_q}$ is constant, i.e., those varying embedding, re-uploadings, batch size, learning rate, or dataset ($17{,}199$ environments), not over the seven axes of Tables~\ref{tab:cit} and \ref{tab:derived}; on the same five axes $1/m_m$ scores $0.7052$ and the single-$c$ composite $0.7012$, while over all seven axes $\mu(c_{A,n_q})$ scores $0.7108$ versus $0.7124$ for $1/m_m$ (paired bootstrap difference $-0.0017$, $95\%$ interval $[-0.0036, 0.0001]$). The gain is confined to the re-uploading axis ($0.680$ versus $0.709$) and is reversed on the qubit ($0.705$ versus $0.685$) and ansatz ($0.771$ versus $0.768$) axes, on which the weights are not fit.

\paragraph{Held-out evaluation of the fitted weights.}
The twelve weights $\{c_{A,n_q}\}$ are fit by minimizing the 90th-percentile sign-error on the same environments on which it is then reported, so the $0.6988$ figure may be optimistic. A seed-level split (fit on seeds $\{16,13\}$, evaluate on $\{42,65,89\}$) is not possible here, because the sign-error requires an effective sample size $n_{\mathrm{eff}}\ge 12$, which needs the full $5\times5$ cross-product of random seeds within each environment; a two- or three-seed subset leaves no scored environments. We therefore split the \emph{environments} in half (preserving all five seeds), fit $\{c_{A,n_q}\}$ on one half, and evaluate the aggregate 90th-percentile sign-error on the other, in both directions. The held-out values are $0.695$ and $0.705$ (versus $0.6988$ in-sample); on the same halves $1/m_m$ scores $0.700$ and $0.709$ and the single-$c$ composite $0.697$ and $0.705$, so the fitted weights keep an advantage of about $0.004$--$0.005$ over $1/m_m$ on these five axes but essentially tie the single-$c$ composite. The single global $c^*$ shows no comparable held-out advantage over $1/m_m$ (held-out composite $0.709$ versus $1/m_m$ $0.710$, and a tie to four decimals ($0.7148$ for both) in the reverse split).

\begin{figure}[h]
    \centering
    \includegraphics[width=0.95\columnwidth]{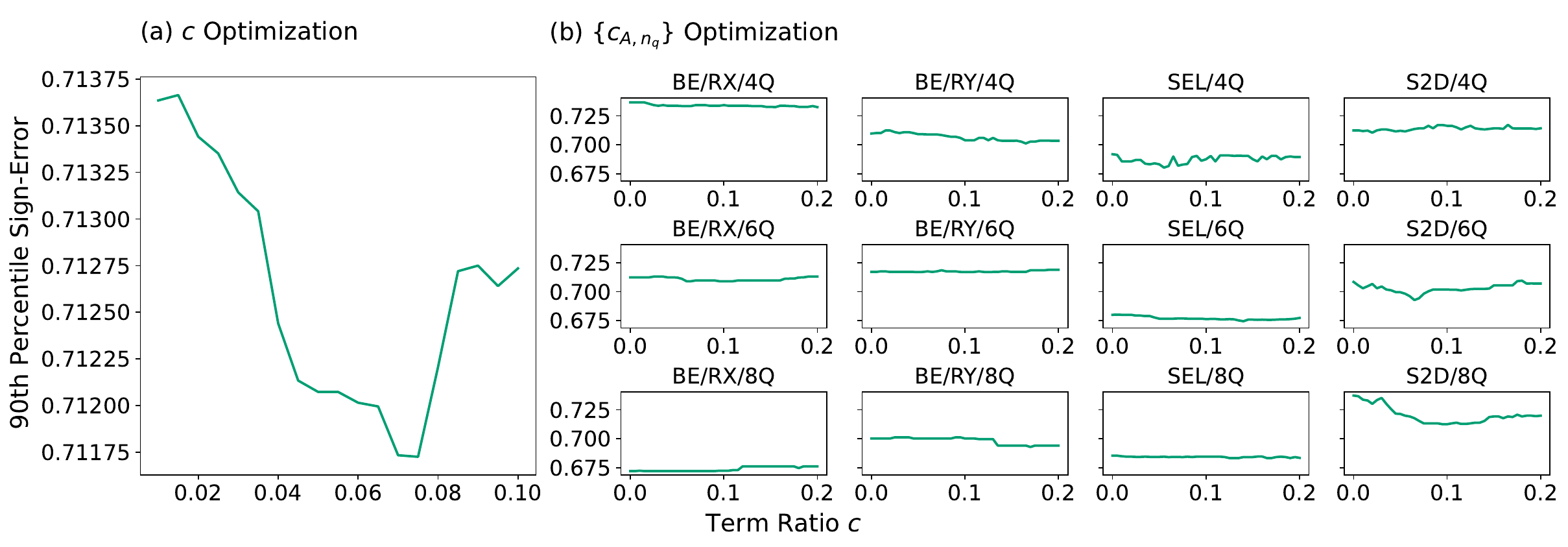}
    \caption{90th-percentile of the empirical sign-error distribution for composite complexity measures $\mu(c) = \sqrt{n_q}/m_m + c\sqrt{n_p}$. (a) One $c$ is used for all experiments. (b) Each (ansatz $A$, number of qubits $n_q$) pair has its own $c_{A,n_q}$. BE: basic entangler; S2D: simplified two-design; SEL: strongly entangling layers. ((a) and (b) have different scales)}
    \label{fig:cscan}
\end{figure}

\section{Margin Statistic Ablation}
\label{app:quantiles}
Table~\ref{tab:quantiles} recomputes the two criteria for all measures using the 10th-percentile and lower-quartile (25th) training margin in place of $m_m$. The $p$-th percentile is exactly zero---and the inverse margin infinite---whenever more than $p\%$ of the training points are misclassified: this occurs for $47.7\%$ ($9{,}279/19{,}440$) of runs at the 10th percentile, $5.9\%$ ($1{,}149$) at the lower quartile, and $0\%$ at the median. The $\cK$ ranking is percentile-independent at both ends: the inverse margin has the highest conditional independence score at every statistic, no measure containing $n_p$ beats it, and the parameter counts sit at the floor. On the sign-error, $n_p/m_m$ ($0.7071$) and the composite ($0.7079$) fall just below $1/m_m$ ($0.7082$) at the 10th percentile, a spread of $0.0011$. Lower quantiles give stronger scores for every margin-based measure---moving from the median to the lower quartile roughly doubles $\cK(1/m_m)$ ($4.57\times10^{-3} \to 1.10\times10^{-2}$) and lowers its 90th-percentile sign-error ($0.7124 \to 0.7017$)---but the 10th-percentile columns rest on a degenerate footing (nearly half the runs at $+\infty$, a third of the sign-error environments lost), so the median is used in the main text.

\begin{table}[h]
\centering
\small
\begin{tabular}{lcccccc}
\toprule
measure & $\cK$ p50 & $\cK$ p25 & $\cK$ p10 & $\widehat{SE}$ p50 & $\widehat{SE}$ p25 & $\widehat{SE}$ p10 \\
\midrule
$1/m_m$ & $4.57$e$\text{-}3$ & $1.10$e$\text{-}2$ & $1.09$e$\text{-}2$ & 0.7124 & 0.7017 & 0.7082 \\
$\sqrt{n_q}/m_m$ & $3.55$e$\text{-}3$ & $9.00$e$\text{-}3$ & $1.05$e$\text{-}2$ & 0.7144 & 0.7044 & 0.7092 \\
$n_p/m_m$ & $4.69$e$\text{-}4$ & $5.86$e$\text{-}3$ & $9.79$e$\text{-}3$ & 0.7256 & 0.7118 & 0.7071 \\
$n_p\sqrt{n_q}$ & $1.57$e$\text{-}4$ & $1.57$e$\text{-}4$ & $1.57$e$\text{-}4$ & 0.7724 & 0.7724 & 0.7724 \\
$n_p$ & $1.94$e$\text{-}4$ & $1.94$e$\text{-}4$ & $1.94$e$\text{-}4$ & 0.7724 & 0.7724 & 0.7724 \\
$n_p\ln(n_p/m_m)$ & $3.66$e$\text{-}4$ & $3.38$e$\text{-}3$ & $5.76$e$\text{-}3$ & 0.7321 & 0.7267 & 0.7237 \\
composite ($c{=}0.075$) & $3.32$e$\text{-}3$ & $8.89$e$\text{-}3$ & $1.06$e$\text{-}2$ & 0.7117 & 0.7028 & 0.7079 \\
\bottomrule
\end{tabular}
\caption{Margin-statistic ablation: $\cK$ (higher is better) and 90\% $\widehat{SE}$ over all environments (lower is better) with the margin taken at the 50th, 25th, and 10th percentile. Runs with $m_m{=}0$: $0$ (p50), $1{,}149$ (p25), $9{,}279$ (p10). The parameter-only measures $n_p$, $n_p\sqrt{n_q}$ do not use the margin, so their scores are constant across statistics.}
\label{tab:quantiles}
\end{table}

\section{Empirical Margin Loss}
\label{app:rgamma}
Table~\ref{tab:rgamma} evaluates the empirical margin loss $\widehat R_\gamma(h) = m^{-1}\sum_i \ind(\cM(h(x_i),y_i) \le \gamma)$ as a complexity measure at three fixed margin levels, with $1/m_m$ for reference. Because negative sample margins are recorded as zero, the count $\#\{\cM \le \gamma\}$ is unchanged by that convention for every $\gamma>0$, so $\widehat R_\gamma$ is the exact empirical margin loss and is finite for every run. It exceeds the inverse median margin on both criteria at all three levels (paired bootstrap of the 90th-percentile difference at $\gamma=0.1$: $-0.0069$, $95\%$ interval $[-0.0111, -0.0034]$; Appendix~\ref{app:ci}); brackets are $95\%$ bootstrap intervals over resampled environments (an independent bootstrap draw from that of Table~\ref{tab:seci}, so the $1/m_m$ interval differs in the last digit).

\begin{table}[h]
\centering
\small
\begin{tabular}{lccc}
\toprule
measure & $\cK(\mu)$ & 90\% $\widehat{SE}$ All [95\% CI] & 90\% $\widehat{SE}$ Common \\
\midrule
$\widehat R_{\gamma=0.05}$ & $1.56\times10^{-2}$ & 0.7059 [0.7011, 0.7100] & 0.7082 \\
$\widehat R_{\gamma=0.1}$  & $1.34\times10^{-2}$ & 0.7056 [0.7011, 0.7090] & 0.7113 \\
$\widehat R_{\gamma=0.2}$  & $9.63\times10^{-3}$ & 0.7082 [0.7042, 0.7122] & 0.7174 \\
$1/m_m$ & $4.57\times10^{-3}$ & 0.7124 [0.7080, 0.7166] & 0.7225 \\
\bottomrule
\end{tabular}
\caption{Empirical margin loss $\widehat R_\gamma$ as a predictor. Higher $\cK$ and lower $\widehat{SE}$ are better.}
\label{tab:rgamma}
\end{table}

\section{Bootstrap Confidence Intervals}
\label{app:ci}
Table~\ref{tab:seci} reports $95\%$ bootstrap intervals for the 90th-percentile sign-error, obtained by resampling the sign-error environments with replacement ($2{,}000$ resamples). The intervals separate the pure-margin measures ($1/m_m$, $\sqrt{n_q}/m_m$, and the composite; upper limit $\le 0.72$) from the parameter-count measures ($n_p$, $n_p\sqrt{n_q}$; lower limit $\ge 0.76$), with $n_p/m_m$ and $n_p\ln(n_p/m_m)$ falling in between, while the differences among the pure-margin measures---including the $0.0007$ edge of the composite over $1/m_m$---lie within the intervals. Because the measures are scored on the same environments, the sharper test is a paired bootstrap of the difference in 90th-percentile sign-error on the shared environments ($24{,}909$; $24{,}897$ for $\widehat R_{\gamma=0.1}$; $2{,}000$ resamples): composite minus $1/m_m$ is $-0.0007$ with $95\%$ interval $[-0.0026, 0.0012]$, and $\widehat R_{\gamma=0.1}$ minus $1/m_m$ is $-0.0069$ $[-0.0111, -0.0034]$ ($-0.0113$ $[-0.0175, -0.0072]$ on the Common environments). Resampling the (varied hyperparameter, base configuration) clusters instead of individual environments leaves these intervals essentially unchanged (composite $[-0.0028, 0.0012]$, $\widehat R_{\gamma=0.1}$ $[-0.0115, -0.0028]$). Confidence intervals for the conditional independence score $\cK$, obtained by resampling the five random seeds with replacement, likewise separate the margin-based from the parameter-count measures (Table~\ref{tab:kci}).


\begin{table}[h]
\centering
\small
\begin{tabular}{lcc}
\toprule
measure & 90\% $\widehat{SE}$ All [95\% CI] & 90\% $\widehat{SE}$ Common [95\% CI] \\
\midrule
$1/m_m$ & 0.7124 [0.7082, 0.7169] & 0.7225 [0.7178, 0.7295] \\
$\sqrt{n_q}/m_m$ & 0.7144 [0.7099, 0.7183] & 0.7274 [0.7216, 0.7342] \\
$n_p/m_m$ & 0.7256 [0.7209, 0.7300] & 0.7531 [0.7468, 0.7594] \\
$n_p$ & 0.7724 [0.7637, 0.7809] & 0.7724 [0.7637, 0.7809] \\
$n_p\sqrt{n_q}$ & 0.7724 [0.7638, 0.7810] & 0.7724 [0.7634, 0.7810] \\
composite ($c{=}0.075$) & 0.7117 [0.7075, 0.7159] & 0.7212 [0.7156, 0.7275] \\
$n_p\ln(n_p/m_m)$ & 0.7321 [0.7275, 0.7367] & 0.7699 [0.7609, 0.7778] \\
\bottomrule
\end{tabular}
\caption{Bootstrap $95\%$ confidence intervals for the 90th-percentile sign-error, resampling environments.}
\label{tab:seci}
\end{table}

\begin{table}[h]
\centering
\small
\begin{tabular}{lc}
\toprule
measure & $\cK(\mu)$ [95\% CI over seeds] \\
\midrule
$1/m_m$ & $4.565\times10^{-3}$ $[3.73,5.18]{\times}10^{-3}$ \\
$\sqrt{n_q}/m_m$ & $3.543\times10^{-3}$ $[2.87,4.33]{\times}10^{-3}$ \\
$n_p/m_m$ & $4.683\times10^{-4}$ $[2.00,7.96]{\times}10^{-4}$ \\
$n_p\sqrt{n_q}$ & $1.564\times10^{-4}$ $[1.56,5.36]{\times}10^{-4}$ \\
$n_p$ & $1.932\times10^{-4}$ $[1.93,5.25]{\times}10^{-4}$ \\
composite ($c{=}0.075$) & $3.316\times10^{-3}$ $[2.69,4.24]{\times}10^{-3}$ \\
$\widehat R_{\gamma=0.1}$ & $1.288\times10^{-2}$ $[1.17,1.54]{\times}10^{-2}$ \\
\bottomrule
\end{tabular}
\caption{Conditional independence score $\cK$ with $95\%$ bootstrap confidence intervals, resampling the five random seeds with replacement ($40$ resamples). With only five seeds the intervals are coarse; the bootstrap medians preserve the ranking of the point estimates, the intervals cleanly separate the margin-based measures from the parameter-count measures, and place $\widehat R_{\gamma=0.1}$ above $1/m_m$. For $n_p$ and $n_p\sqrt{n_q}$ no resample falls below the full-grid value of the bootstrap run, which therefore sits at the lower limit.}
\label{tab:kci}
\end{table}

\clearpage
\section{Sign-Error Empirical Cumulative Distributions}
\label{app:se_ecdf}

\begin{figure}[h]
    \centering
    \includegraphics[width=0.5\columnwidth]{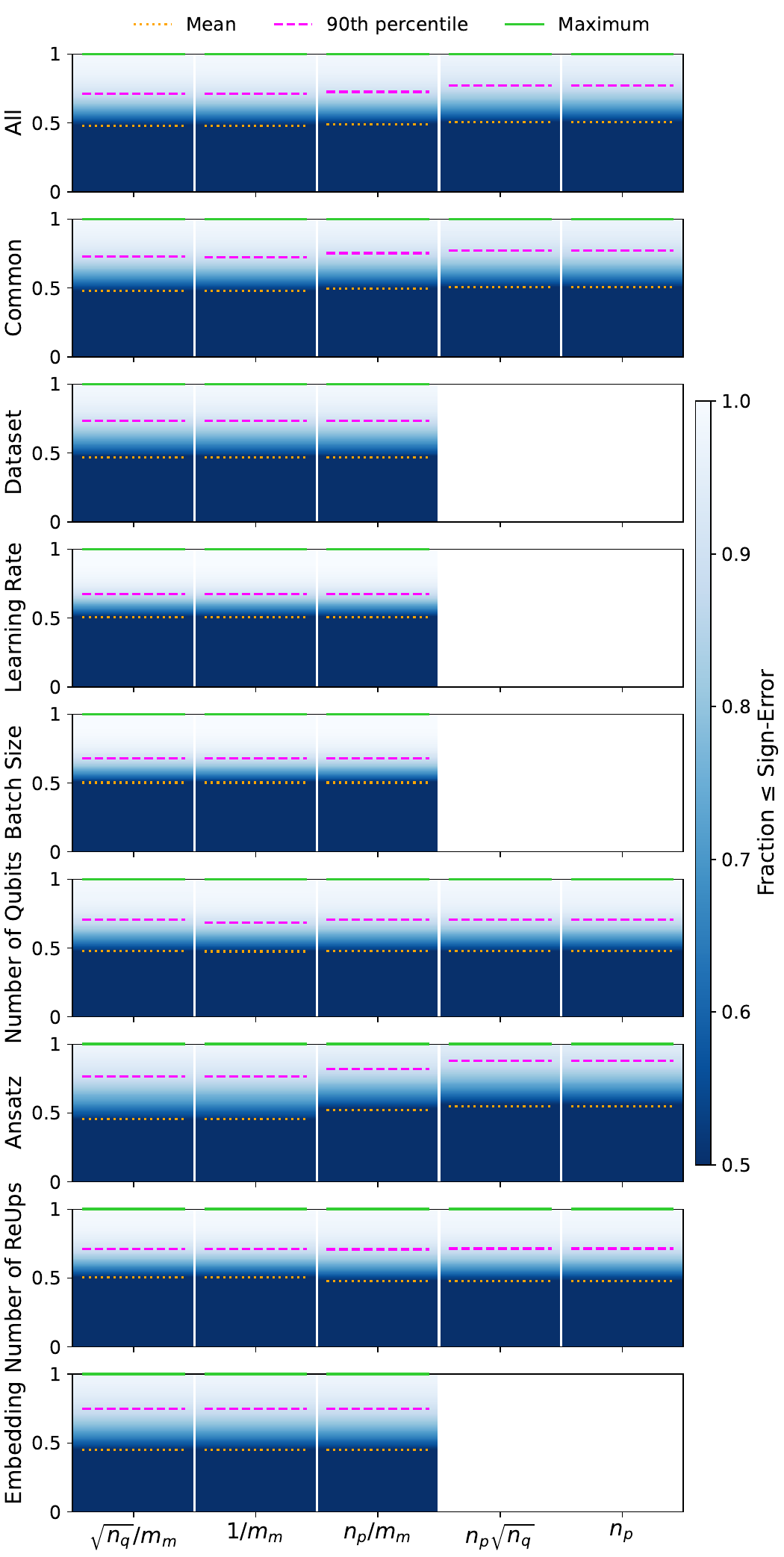}
    \caption{Empirical cumulative distributions of sign-error for the different complexity measures considered in this study. The mean, 90th-percentile, and maximum values are indicated with horizontal lines.}
    \label{fig:se_ecdf}
\end{figure}

\end{document}